\documentclass[aps,pra,superscriptaddress,10pt,floatfix]{revtex4-2} 

\usepackage{amsmath}
\usepackage{amssymb}
\usepackage{amsthm}
\usepackage{amsfonts}
\usepackage{bbm}
\usepackage{xcolor}
\usepackage{graphicx} 
\usepackage[utf8]{inputenc}
\usepackage[T1]{fontenc}
\usepackage{physics}
\usepackage{times}
\usepackage{cancel}
\usepackage{wasysym}
\usepackage{comment}
\usepackage{ulem}
\usepackage{mathtools}
\usepackage{lipsum}
\usepackage{placeins}

\usepackage[caption=false]{subfig}

\usepackage{hyperref}

\hypersetup{
	colorlinks=true,linkcolor=blue,citecolor=blue,
	filecolor=blue,urlcolor=blue,breaklinks=true
}

\theoremstyle{remark}
\newtheorem*{remark}{Remark}
\theoremstyle{plain} 
\newtheorem{theorem}{Theorem}[section]
\newtheorem{lemma}[theorem]{Lemma}
\newtheorem{corollary}[theorem]{Corollary}

\newcommand{\cH}{\mathcal{H}}
\newcommand{\B}{\mathcal{B}}

\hypersetup{
	colorlinks=true,linkcolor=blue,citecolor=blue,
	filecolor=blue,urlcolor=blue,breaklinks=true
}

\newcommand{\orcid}[1]{\href{https://orcid.org/#1}{\includegraphics[width=7pt]{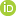}}}

\begin{document}
	
\title{Effective Transition from Weak to Essential Non-Markovianity Induced by Coarse-Graining}

\author{Gabriel M. Arantes\orcid{0009-0004-7382-2635}}
\email{gabriel.moniz@usp.br}
\affiliation{
	Department of Mathematical Physics, Institute of Physics, University of São Paulo, Rua do Matão 1371, São Paulo 05508-090, São Paulo, Brazil
}
\author{Barbara Amaral \orcid{0000-0003-1187-3643}}
\email{barbara_amaral@usp.br}
\affiliation{
	Department of Mathematical Physics, Institute of Physics, University of São Paulo, Rua do Matão 1371, São Paulo 05508-090, São Paulo, Brazil
}

\author{Nadja K. Bernardes \orcid{0000-0001-6307-411X}}
\email{nadja.bernardes@ufpe.br}
\affiliation{
	Departamento de F\'isica, Centro de Ci\^encias Exatas e da Natureza, Universidade Federal de Pernambuco, Recife 50670-901, Brazil
}

\date{\today}
	
\begin{abstract}
Quantum channels generally reduce the distinguishability of quantum states, thereby constraining information transmission and processing in open quantum systems. While it is known that distinguishability can be partially recovered through suitable post-processing protocols, a systematic characterization of the maximal achievable gain has remained elusive. Here, we establish a general framework to determine and optimize the recovery of distinguishability induced by a quantum channel. We introduce an algorithm that identifies the optimal implementation of a multi-copy coarse-graining protocol applicable to arbitrary channels. Within this framework, we derive a general upper bound on the attainable distinguishability gain and quantify the performance of the protocol through its tightness relative to this bound. Our results show that collective coarse-graining can lead to effective dynamics exhibiting the operational signatures of essential non-Markovianity even when the underlying microscopic dynamics remain weakly non-Markovian. A detailed analysis reveals a nontrivial trade-off between mathematical optimality (bound saturation) and operational optimality, together with a strong dependence on both the input ensemble and the number of copies. Taken together, these findings provide a unified and quantitative framework to assess, optimize, and interpret distinguishability recovery in open quantum systems while showing how collective processing can qualitatively modify the operational memory properties of effective dynamical descriptions. 
\end{abstract}

\maketitle

\section{Introduction}
\label{introduction}

The dynamics of open quantum systems are fundamentally shaped by the exchange of information between a system and its environment \cite{Breuer2009, Rivas2014}. A central distinction in this context is between Markovian and non-Markovian processes \cite{Rivas2014}. Markovian dynamics are commonly associated with a monotonic, though not necessarily strict, loss of system-accessible information, typically formalized through completely positive (CP) divisibility \cite{Rivas2010}. In contrast, non-Markovian processes violate this property and may exhibit memory effects, most notably in the form of information backflow, whereby previously lost information temporarily returns to the system \cite{Breuer2009, Rivas2014, ConstructiveBackflow}.

Beyond this binary classification, non-Markovianity admits a finer operational hierarchy. In particular, one can distinguish between \emph{weakly} and \emph{essentially} (or strong) non-Markovian dynamics, depending on the positivity properties of the intermediate dynamical maps~\cite{weak-strong}. While both regimes violate CP-divisibility \cite{Rivas2010}, only essentially non-Markovian processes exhibit a recovery of distinguishability between states of the system alone, typically witnessed by a temporary increase in the trace distance between quantum states \cite{Breuer2009, weak-strong}. Weakly non-Markovian dynamics, by contrast, remain operationally indistinguishable from Markovian ones under the standard distinguishability measure, as they require an extension of the channel by an ancilla space to observable recovery of information \cite{ConstructiveBackflow}.

This distinction has direct implications for quantum information processing, as information backflow is a resource in tasks such as quantum teleportation~\cite{Teleportation}, error correction~\cite{errorcorrection}, entanglement preservation~\cite{Entanglement-preservation}, and thermodynamics~\cite{NM-thermodynamics}. Consequently, only essentially non-Markovian dynamics carry the operational signatures required to power these tasks. This naturally raises the question of whether collective processing can generate such signatures from a weakly non-Markovian process without altering the underlying microscopic dynamics.

The possibility of enhancing or altering non-Markovian memory effects through system interactions and collective operations has recently attracted significant attention in open quantum systems. Multi-copy processing and coarse-graining operations have been shown to manipulate quantum information flows across parallel implementations of dynamical maps. In particular, Azevedo et al. \cite{AZEVEDO} introduced a multi-copy coarse-graining protocol analogous to concepts in quantum resource theories \cite{Nonmarkovianitasresource}, capable of reshaping the effective memory properties of quantum channels. Similar phenomena involving collective processing, superactivation, and emergent dynamics have been widely investigated in quantum information science, where access to higher tensor powers ($n \ge 2$) unlocks operational capacities that remain strictly forbidden at the single-copy level \cite{Bennett1996, Horodecki1999, Smith2008}. In addition to post-processing protocols, physical mechanisms such as environment-induced indirect interactions among multi-component systems have also been shown to amplify non-Markovian memory signatures by orders of magnitude \cite{Zaman2026}. However, previous multi-copy coarse-graining schemes revealed a fundamental bottleneck: while effective in the strongly non-Markovian regime, they failed to induce observable information backflow when applied to weakly non-Markovian channels \cite{AZEVEDO}. Whether this failure represents an intrinsic physical limitation or merely a consequence of suboptimal coarse-graining maps has  remained an open question.

In this work, we answer this question in the affirmative by optimizing the multi-copy protocol for quantum processes presented in reference \cite{AZEVEDO}. The protocol operates by coarse-graining  multiple independent instances of a given dynamical map, thereby effectively reshaping its information flow properties. We show that this procedure can induce an effective transition from weak to essential non-Markovianity, as witnessed by the emergence of a positive increase in distinguishability. 

To quantify the performance of this transformation, we derive a general upper bound on the achievable increase in distinguishability and introduce an associated tightness measure that captures how closely the protocol approaches this limit. This framework allows us to systematically analyze the interplay between bound saturation and operational gain, revealing a nontrivial trade-off: scenarios that saturate the bound do not necessarily correspond to those with maximal performance.

Furthermore, we demonstrate that the effectiveness of the protocol depends sensitively on both the structure of the input state ensemble and the number of copies employed. While increasing the number of copies generally enhances the distinguishability change, we identify regimes in which this improvement is strongly suppressed. Our analysis also uncovers parameter regions where the protocol achieves the emergence of operational signatures despite extremely small initial distinguishability, highlighting both the power and the limitations of collective coarse-graining.

\section{The Model and Non-Markovianity Regimes}

We consider a finite-dimensional quantum system with Hilbert space $\mathcal H = \mathbb{C}^d$. States are represented by density operators $\rho \in D(\mathcal H)$, and their dynamics is described by a family of completely positive and trace-preserving (CPTP) maps $\Lambda = \{\Lambda_t\}_{t \geq 0}$ with $\Lambda_0 = \mathbb{I}$.

A central structural property of open quantum dynamics is Markovianity, commonly defined in terms of CP-divisibility \cite{Rivas2010}. A quantum evolution is said to be Markovian if, for all $t \geq s \geq 0$, it can be decomposed as
\begin{equation}
\Lambda_t = V_{t,s} \circ \Lambda_s,
\end{equation}
where the intermediate maps $V_{t,s}$ are CPTP. Any deviation from this condition signals non-Markovian behavior.

Non-Markovian dynamics, however, admit a finer hierarchy based on the positivity properties of the intermediate maps~\cite{weak-strong}. In particular, one distinguishes between \emph{weakly} and \emph{essentially} non-Markovian evolutions. Weak non-Markovianity corresponds to the case in which $V_{t,s}$ remains positive (P-divisible) but not completely positive, whereas essential non-Markovianity arises when positivity itself is violated. This distinction defines a sharp operational boundary: only essentially non-Markovian processes can exhibit a genuine recovery of information accessible to the system \cite{weak-strong}.

This operational difference is naturally captured by the behavior of the trace distance between quantum states. The trace distance provides a direct, operational quantification of the physical distinguishability between two quantum states. For $\rho_1,\rho_2 \in D(\mathcal H)$, the trace distance is defined as
\[
D(\rho_1,\rho_2) := \tfrac{1}{2}\|\rho_1 - \rho_2\|_1.
\] 
The trace distance directly dictates the maximum success probability achievable when attempting to distinguish between the states in a single measurement, given by $$P_{\text{succ}} = \frac{1}{2}\left(1 + D(\rho_1, \rho_2)\right)$$ via the Helstrom bound \cite{Helstrom1969, Helstrom1976, Barnett2009}. When $D(\rho_1, \rho_2) = 1$, the states have orthogonal supports and are perfectly distinguishable; conversely, when $D(\rho_1, \rho_2) = 0$, the states are identical and operationally indistinguishable. Furthermore, because any completely positive and trace-preserving (CPTP) quantum channel $\Lambda$ is contractive under the trace norm—satisfying $$D(\Lambda(\rho_1), \Lambda(\rho_2)) \le D(\rho_1, \rho_2).$$ The trace distance formalizes the principle that open-system dynamics and quantum processing operations cannot increase the distinguishability of states without non-Markovian memory effects or collective multi-copy processing.

In a discrete-time setting, considering two successive instants, we define the change in distinguishability as
\begin{equation}
\Delta D(\rho_1,\rho_2)
=
\frac{1}{2}\big(
\|\Lambda_2(\rho_2 - \rho_1)\|_1
-
\|\Lambda_1(\rho_2 - \rho_1)\|_1
\big).
\label{eq:deltaD}
\end{equation}
If the intermediate map $\Lambda_{2,1} = \Lambda_2 \circ \Lambda_1^{-1}$ is positive, contractivity of the trace norm implies $\Delta D \leq 0$ for all input states (see Appendix~\ref{app: distinguishability as witness of the Regime}). This establishes that the change of distinguishability between states in $D(\mathcal{H})$ can only witness a violation of P-divisibility, whereas to witness a violation of CP-divisibility it is necessary to consider the corresponding extension of the dynamical map to the system with an ancilla \cite{ConstructiveBackflow}. 

A positive value of $\Delta D$ therefore provides an operational witness of essential non-Markovianity, signaling information backflow from the environment to the system. In contrast, weakly non-Markovian dynamics remain indistinguishable from Markovian ones under this criterion, as they do not generate any observable recovery of distinguishability.

This limitation motivates the introduction of dynamical transformations capable of modifying the effective information flow of a process. In particular, we consider transformations acting on quantum evolutions through multi-copy processing and coarse graining. Given $n$ independent realizations of a channel $\Lambda_t$, we define the effective dynamics
\begin{equation}
\Lambda'_t = \lambda \circ \big(\Lambda_t^{\otimes n}\big),
\label{effective channel}
\end{equation}
where $\lambda$ is a CPTP map that compresses the $n$-copy output back to the original system dimension. As discussed in the previous section, this construction preserves complete positivity but can fundamentally alter the non-Markovian character of the dynamics by reshaping correlations across copies.

The possibility of enhancing non-Markovianity through such transformations was first explored in \cite{AZEVEDO}, where the process was described as a form of non-Markovianity distillation. Nevertheless, within the resource-theoretic framework of \cite{Nonmarkovianitasresource}, such an interpretation is not immediate. There, Markovian processes define the free operations, while non-Markovian dynamics are regarded as resourceful. Consequently, the observed enhancement of non-Markovianity is more naturally interpreted as a resource-generating transformation unless an alternative set of free operations is specified.

Despite this, the protocol reveals interesting physical features. In particular, the coarse-graining procedure can correlate different copies of the dynamics in such a way that the distinguishability of quantum states increases over time. As a result, an observer having access only to a memoryless weakly non-Markovian channel may engineer effective dynamics exhibiting a genuine backflow of information. From this perspective, the protocol is most naturally interpreted as a mechanism for generating effective memory through collective processing. Indeed, the resulting non-Markovianity generally originates from the coarse-graining transformation itself rather than from the concentration or amplification of memory effects already present in the original dynamics.

However, an important limitation was identified in \cite{AZEVEDO}: while the protocol is effective in the essentially non-Markovian regime, it fails to induce any observable backflow in the weak regime. Since weakly non-Markovian processes satisfy $\Delta D \leq 0$ for all input states, it remains unclear whether this limitation is fundamental or merely a consequence of suboptimal processing.

In this work, we address this question by systematically optimizing over all admissible coarse-graining maps $\lambda$, as detailed in the protocol section. Our goal is to determine whether multi-copy processing can activate essential non-Markovianity, i.e., induce $\Delta D'_n > 0$, even when the original dynamics is confined to the weak regime. In addition, we quantify the maximal achievable enhancement and analyze its dependence on the number of copies and the structure of the input states.

To explicitly investigate these questions, we consider a discrete-time model defined by two successive quantum channels,
\begin{equation}
\begin{split}
    \Lambda_1(\rho) &=(1-2\varepsilon)\rho + \varepsilon(Z\rho Z + X\rho X), \\
    \Lambda_2(\rho) &= \big[(1-2\varepsilon)^2 + 4\varepsilon^2\big]\rho 
    + 2\varepsilon(1-2\varepsilon)(Z\rho Z + X \rho X),
\end{split}
\label{eq:model}
\end{equation}
previously analyzed in Ref.~\cite{Weak-Model}. For this model, the intermediate map $\Lambda_{2,1}$ is positive but not completely positive in the parameter range $0 < \varepsilon < 0.25$, corresponding to the weakly non-Markovian regime. This makes it an ideal testbed to investigate whether multi-copy effective dynamics can induce a effective transition to essential non-Markovianity.

As will be shown, the interplay between the structure of the input ensemble, the number of copies, and the fundamental bounds governing distinguishability leads to a rich phenomenology, including multi-copy activation thresholds and nontrivial trade-offs between performance and bound saturation.

\section{Definitions and Theorems}
\label{sec:preliminaries}

The goal of the protocol is to enhance the observable distinguishability between different stages of the dynamics through an optimized effective description. Operationally, this requires transforming the operators associated with the evolution in such a way that their trace norms become as separated as possible after coarse-graining.

However, two independent mechanisms constrain this enhancement. First, the trace norm obeys a reverse-triangle inequality, which limits how much the norms of two operators can differ relative to their mutual distance. Second, CPTP maps cannot increase distinguishability, imposing a fundamental contractivity constraint on the admissible transformations.

For any Hermitian operators $A,B \in \mathcal{B}(\cH_m)$, these constraints combine into the chain of inequalities
\begin{align}
    \big|\ \|\lambda(A)\|_1 - \|\lambda(B)\|_1 \ \big| 
    &\leq  \|\lambda(A-B)\|_1 \label{eq:ineq1}\\
    &\leq \|A-B\|_1, \label{eq:ineq2}
\end{align}
where Eq.~\eqref{eq:ineq1} follows from the triangle inequality and Eq.~\eqref{eq:ineq2} from the contractivity of the trace norm under CPTP maps. This relation establishes a fundamental upper bound on the distinguishability gain achievable through any coarse-graining procedure.

Saturating the bound, therefore, requires simultaneously overcoming two distinct obstructions established by the following lemmas:
\begin{lemma}
Let $A, B \in \mathcal{B}(H_m)$ be traceless Hermitian matrices. Let $\Delta = A - B$ and $P_+$ be the projector onto the strictly positive eigenspace of $\Delta$.
\[
 \|A\|_1 - \|B\|_1 = \|A-B\|_1
\]
if and only if $A$ and $B$ commute and, for every common eigenvector $|v\rangle$ with 
\[
A|v\rangle = v_a |v\rangle, \qquad B|v\rangle = v_b |v\rangle,
\]
the corresponding eigenvalues satisfy $\text{Tr}(P_+B) \geq 0$.
\label{cond 1}
\end{lemma}

\begin{lemma}
Given $\Delta \in \mathcal B(H_m)$ a Hermitian matrix and a CPTP map $\lambda:  \mathcal B(H_m) \rightarrow  \mathcal B(H_d)$. Take the spectral decompositions $\Delta= \Delta_+ - \Delta_-$ and $\lambda(\Delta) = Q_+ - Q_-$. Let $\Pi_{+}$ and $\Pi_{-}$ be the projectors onto the supports of $Q_+$ and $Q_-$, respectively. We have 
\begin{equation*}
    \|\lambda_U(\Delta)\|_1 = \|\Delta\|_1
\end{equation*}
if and only if $\Pi_{\pm}\lambda(\Delta_{\pm})\Pi_{\pm}=\lambda(\Delta_{\pm})$ and $\Pi_{\pm}\lambda(\Delta_{\mp})\Pi_{\pm}=0$.

\end{lemma}

The first equality demands that the coarse-grained operators become spectrally aligned, so that their contributions to the trace norm add constructively rather than destructively. The second requires that the CPTP map preserve the full distinguishability contained in $A-B$, which is only possible when the positive and negative sectors of the operator remain perfectly separated under the effective dynamics. Detailed proofs of these conditions are provided in Appendix~\ref{app: Conditions Lemma}.

Understanding when this bound can be saturated is essential for distinguishing intrinsic limitations of the dynamics from limitations arising purely from suboptimal multi-copy post-processing.

While each of these conditions can be satisfied individually by suitable choices of $\lambda$, their simultaneous fulfillment is nontrivial, where the existence of optimal CPTP may not be guaranteed. The following theorem characterizes precisely when this is possible for the class of traceless Hermitian operators relevant to our setting.

\begin{theorem}
Let $A, B \in \mathcal{B}(\cH_m)$ be traceless Hermitian operators, and define $\Delta = A - B$. Let $P_+$ and $P_-$ denote the projectors onto the strictly positive and strictly negative eigenspaces of $\Delta$, respectively, and $P_0 = I - P_+ - P_-$ the projector onto its kernel. There exists a CPTP map $\lambda: \mathcal{B}(\cH_m) \to \mathcal{B}(\cH_d)$, with $d \geq 2$, such that
\begin{equation}
    \|\lambda(A)\|_1 - \|\lambda(B)\|_1 = \|\lambda(A-B)\|_1 = \|A-B\|_1
    \label{eq:saturation_condition}
\end{equation}
if and only if there exists an operator $0 \leq E_0 \leq P_0$ satisfying
\[
\operatorname{Tr}(P_+B) + \operatorname{Tr}(E_0 B)
\;\geq\; 0.
\]
\label{Theorem saturating states}
\end{theorem}

\begin{proof}
See Appendix~\ref{app: general bound theorem}.
\end{proof}

Physically, the theorem shows that maximal distinguishability enhancement is only possible when the coarse-graining map can simultaneously isolate the positive and negative sectors of the operator difference $\Delta= A- B$ while preserving a compatible ordering between the spectral contributions of A and B. The kernel contribution encoded in $E_0$ reflects a residual freedom in how the null subspace of $\Delta$ is processed, revealing that saturation depends not only on the magnitude of the distinguishability imbalance, but also on its detailed spectral organization.

More generally, the theorem identifies a fundamental geometric obstruction in achieving the bound. Although distinguishability can always be partially concentrated through collective processing, achieving the absolute upper bound requires a highly constrained alignment between the spectral structures of the evolved operators.

In the context of the protocol, we are interested in operators of the form
\begin{equation}
A^{(n)} = \Lambda_1^{\otimes n}(\rho_1^{\otimes n} - \rho_2^{\otimes n}),
\quad
B^{(n)} = \Lambda_2^{\otimes n}(\rho_1^{\otimes n} - \rho_2^{\otimes n}),
\label{matrix forms}
\end{equation}

for some integer $n$. The theorem then implies that, once the underlying dynamics is fixed, the attainability of the general bound is entirely determined by the choice of input states. This observation underlies the ensemble-dependent behavior observed in the numerical analysis, as well as the trade-off between bound saturation and operational gain discussed in later sections.

Another important observation is that the upper bound \eqref{eq:ineq2} calculated for the operators \eqref{matrix forms} is non-decreasing with $n$. Indeed, the trace norm is contractive under the partial trace and $\Tr_1(A^{(n+1)}-B^{(n+1)})=A^{(n)}-B^{(n)}$, thus obtaining: $$||A^{(n+1)}-B^{(n+1)}||_1\geq||A^{(n)}-B^{(n)}||_1.$$

Finally, to implement the optimization over $\lambda$, we recall a Stinespring-type dilation theorem, which provides a complete parametrization of CPTP maps. Let $\cH_m \cong \mathbb{C}^{m}$ denote a finite-dimensional Hilbert space, and $\B(\cH_m)$ the algebra of linear operators acting on it.

\begin{theorem}[Stinespring-type dilation {\cite[Th.~2]{Nadja-proof}}]
Let $\lambda:\mathcal B(\mathcal H_m)\to\mathcal B(\mathcal H_d)$ be a CPTP map. Then there exists an auxiliary Hilbert space $\mathcal H_r$, with $\dim \mathcal H_r = r \le d$, and a unitary 
\[
U \in \mathcal U(\mathcal H_m \otimes \mathcal H_r \otimes \mathcal H_d)
\]
such that, for all $\psi \in \mathcal B(\mathcal H_m)$,
\begin{equation}
\lambda(\psi)
=
\operatorname{Tr}_{m,r}\!\left[
U\big(
\psi \otimes |0_r\rangle\langle 0_r| \otimes |0_d\rangle\langle 0_d|
\big)U^\dagger
\right].
\label{eq:stinespring}
\end{equation}
\end{theorem}

This representation is central to our numerical approach, as it allows the optimization over CPTP maps to be recast as an optimization over unitary operators. In particular, it provides a parametrization of the coarse-graining map in terms of $(mrd)^2$ real parameters via the Lie algebra of $\mathcal U(\mathcal H_m \otimes \mathcal H_r \otimes \mathcal H_d)$, enabling a systematic optimization. 

The protocol requires that for $n$ copies, the dimension of the input space is $m=2^n$, so the dimension increases exponentially, rendering the optimization computationally demanding even for moderate $n$. However, as we will see in the following sections, the cases $n=2$ and $n=3$ already capture the relevant qualitative behaviors.

\section{Effective Dynamics}
\label{sec:distillation}

The effective dynamics introduced in Ref.~\cite{AZEVEDO}, illustrated schematically in Fig.~\ref{fig:circuit}, aims to enhance the distinguishability change of quantum states after the action of a quantum channel by exploiting multi-copy processing. The central idea is to combine multiple independent realizations of a dynamical map and extract, via a suitable coarse-graining procedure, an effective dynamics with improved information-theoretic properties.

We consider $n$ identical copies of both the quantum channel and the input states. After the action of the channel, a coarse-graining map $\lambda$ is applied to the resulting $n$-fold tensor-product state, yielding an effective single-qubit quantum operation
\begin{equation}
    \Lambda'_i(\rho) = \lambda\!\left[\Lambda_i(\rho)^{\otimes n}\right],
    \label{effective dynamics}
\end{equation}
where the output space has the same dimension as the original system. While this construction preserves the system dimension, it enables the generation and processing of nontrivial correlations across copies, effectively reshaping the information flow of the dynamics.

\begin{figure}[h]
    \centering
    \includegraphics[width=0.5\linewidth]{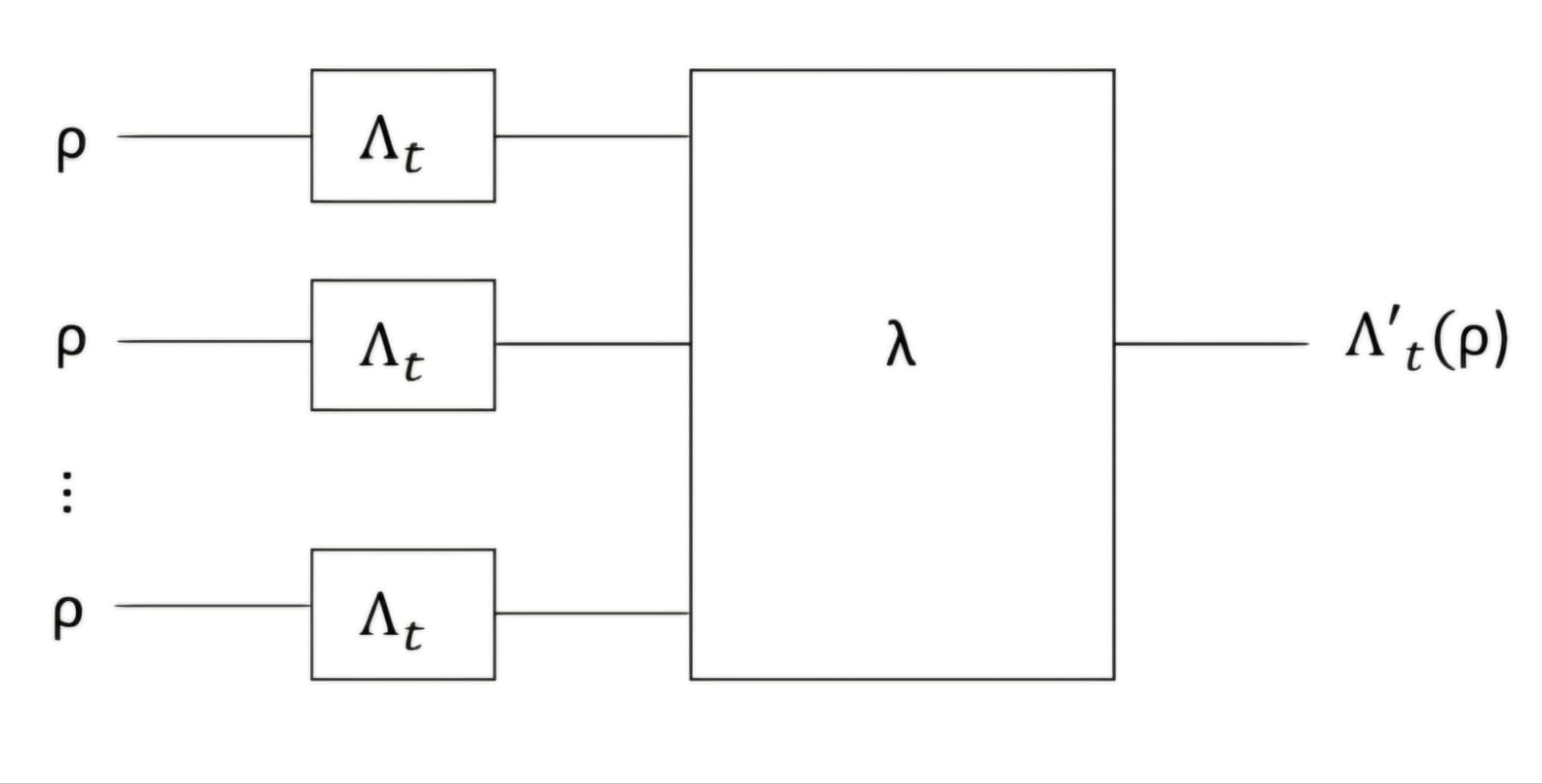}
    \caption{Schematic representation of the non-Markovianity protocol. \cite{AZEVEDO}}
    \label{fig:circuit}
\end{figure}

Formally, let $\mathcal H_A$ denote the Hilbert space of the original system, $\mathcal H_B$ an auxiliary space introduced by the coarse-graining, and $\mathcal H_C$ an ancilla space defining the output system. The map $\lambda$ is required to be completely positive and trace preserving, and can therefore be expressed via a Stinespring dilation as
\begin{equation}
    \lambda_U(\psi)
    =
    \operatorname{Tr}_{A,B}\!\left[
        U\!\left(
            \psi \otimes |0_B\rangle\!\langle 0_B|
            \otimes |0_C\rangle\!\langle 0_C|
        \right)
        U^\dagger
    \right],
    \label{coarse-graining}
\end{equation}
for a unitary
\[
U \in \mathcal U\!\left(\mathcal H_A^{\otimes n} \otimes \mathcal H_B \otimes \mathcal H_C\right).
\]

The operational figure of merit is the change in distinguishability induced by the effective dynamics. For a pair of input states $\rho_1$ and $\rho_2$, we define
\begin{equation}
    \Delta D_n'(\rho_1,\rho_2,U)
    = \frac{1}{2} \big(
    \bigl\| \Lambda'_2(\rho_2 - \rho_1) \bigr\|_1
    -
    \bigl\| \Lambda'_1(\rho_2 - \rho_1) \bigr\|_1\big),
    \label{eq:distilled_change}
\end{equation}
which quantifies the optimized distinguishability gain. The protocol seeks to maximize this quantity over $U$, with particular emphasis on achieving $\Delta D'_n > 0$, thereby inducing information backflow and promoting the dynamics to the essentially non-Markovian regime.

For fixed input states, number of copies $n$, and auxiliary dimension, the optimization reduces to a search over the unitary $U$ in Eq.~\eqref{coarse-graining}. However, this optimization is fundamentally constrained. Building on the framework established in Sec.~\ref{sec:preliminaries}, the achievable gain satisfies the upper bound
\begin{equation}
\Delta D_n'(\rho_1,\rho_2,U) \leq\left\lVert 
    \Lambda_2(\rho_2)^{\otimes n} - \Lambda_2(\rho_1)^{\otimes n}
    - \Lambda_1(\rho_2)^{\otimes n} + \Lambda_1(\rho_1)^{\otimes n} 
    \right\rVert_1.
\label{ineq: upper bound protocol}
\end{equation}
Independently, the trace distance imposes the structural constraint $\Delta D'_n \leq 1$. To ensure a physically meaningful benchmark, we therefore define the \emph{effective general bound}
\begin{equation}
    \beta_n(\varepsilon) = \min\left\{ 
    1, \right.
    \left.\left\lVert \Lambda_2(\rho_2)^{\otimes n} - \Lambda_2(\rho_1)^{\otimes n}
    - \Lambda_1(\rho_2)^{\otimes n} + \Lambda_1(\rho_1)^{\otimes n} 
    \right\rVert_1     \right\}. 
\label{eq:general_bound}
\end{equation}

The performance of the protocol is assessed by comparing the optimized value of $\Delta D'_n$ with both the original baseline $\Delta D$ and the bound $\beta_n$. This comparison allows us to quantify three distinct aspects: the absolute gain $\Delta D'_n - \Delta D$, the efficiency of the protocol relative to its fundamental limit, and the conditions under which bound saturation occurs. As will be shown, these features are not trivially aligned: in particular, maximal gain does not necessarily coincide with saturation of $\beta_n$, revealing a nontrivial trade-off between operational advantage and proximity to the bound. 

To gain analytical insight into these mechanisms, it is instructive to consider a simple yet representative pair of states, $\rho_1=|1\rangle\langle1|$ and $\rho_2=|0\rangle\langle0|$. Under the channel model, their single-copy evolution reads
\begin{equation}
\begin{split} 
    \Lambda_1(\rho_i) &= \frac{1}{2}\big[ I + (-1)^i(1 -2\varepsilon)Z\big],\\
    \Lambda_2(\rho_i) &= \frac{1}{2}\big[ I + (-1)^i(1 -4\varepsilon + 8\varepsilon^2)Z\big].
\end{split}
\label{eq:analytical_case}
\end{equation}
The corresponding original change is $\Delta D = |1-4\varepsilon+8\varepsilon^2| - |1-2\varepsilon|$. In the strongly non-Markovian regime ($0.25 \le \varepsilon \le 0.5$), this pair saturates the bound for $n=2$ already at the single-copy level, implying that no improvement is possible at the two-copies level.

Extending to $n$ copies, one can construct the operators
\begin{widetext}
    \begin{equation}
\begin{split}
    B &\coloneqq \Lambda_1^{\otimes n}(\rho_2^{\otimes n}-\rho_1^{\otimes n}) = \frac{1}{2^{n-1}}\bigg[ \sum_{k=1}^{m} (1 -2\varepsilon)^{2k-1} \sum_{j_1,...,j_{2k-1}}\prod_{l=1}^{{2k-1}} Z^{(j_l)}\bigg], \\
    A &\coloneqq \Lambda_2^{\otimes n}(\rho_2^{\otimes n}-\rho_1^{\otimes n}) = \frac{1}{2^{n-1}}\bigg[\sum_{k=1}^{m} (1 -4\varepsilon + 8\varepsilon^2)^{2k-1}\sum_{j_1,...,j_{2k-1}}\prod_{l=1}^{{2k-1}} Z^{(j_l)}\bigg].
\end{split}
\end{equation}
\end{widetext}
which admit an expansion in odd-body Pauli operators. In the strong regime, the eigenvalue ordering condition required by Theorem~\ref{Theorem saturating states} is satisfied, ensuring that the bound $\beta_n$ is attained. By contrast, in the weakly non-Markovian region ($\varepsilon < 0.25$), this ordering is reversed, and the analytical guarantee of saturation is lost.

This distinction has direct operational consequences. While $\Delta D$ is negative in the weak regime, the key question is whether the multi-copy protocol can nevertheless give $\Delta D'_n > 0$. Analytical methods alone are insufficient to resolve this, necessitating numerical optimization. As shown in Appendix~\ref{app: data trivial case}, the protocol exhibits a clear multi-copy threshold: for $n=2$, no regime activation occurs ($\Delta D'_2 = 0$), whereas for $n=3$, one obtains $\Delta D'_3 > 0$, signaling the operational change to essential non-Markovianity. In the strong regime, the bound is saturated, but no additional gain is observed for two copies and a small gain for three copies.

This example highlights two general features that will be explored in detail in the following sections. First, the effectiveness of the protocol depends sensitively on the parameter regime and the structure of the input ensemble. Second, increasing the number of copies can qualitatively alter the dynamics, enabling a change that is inaccessible at a lower number of copies. These observations motivate the comprehensive numerical analysis presented below.

Finally, we note that the qualitative behavior described here is not specific to the computational basis: choosing, for instance, the eigenstates of $X$ leads to an analogous structure, with identical conclusions.

\section{Results}
\label{sec:results}

To systematically evaluate the performance of the optimization, our numerical analysis proceeds in two stages. First, we establish the global qualitative behavior of the protocol across a broad landscape of initial states, analyzing ensembles of mixed, random pure, and orthogonal pure states. Having mapped the general statistical trends of bound saturation and multi-copy advantage, we next isolate the physical extremes of the protocol.

For each ensemble, we extract two representative pairs of states that delimit the onset of the effective transition. Specifically, the worst-case scenario is defined as the state pair exhibiting the smallest peak distinguishability change in the weakly non-Markovian region, formally given by
\[
\min\left[\max_{0<\varepsilon<0.25}\Delta D'_n(\varepsilon)\right].
\]
Conversely, the best-case scenario corresponds to the pair attaining the largest peak distinguishability change,
\[
\max\left[\max_{0<\varepsilon<0.25}\Delta D'_n(\varepsilon)\right].
\]

For the heatmaps, to enable a coherent visual comparison across different initial states and metrics, the vertical axis in all heatmaps corresponds to the randomly sampled state pairs, sorted by a uniquely defined global metric. 

Let $k \in \{1, 2, \dots, N\}$ denote the index of a state pair in a given ensemble, where $N$ is the total number of sampled pairs. We define a sorting score, $S_k$, based on the average tightness ratio achieved by the $k$-th pair at the two-copy level ($n=2$) strictly within the strongly non-Markovian regime ($\epsilon > 0.25$):
\begin{equation}
    S_k = \frac{1}{|\mathcal{E}_{\text{strong}}|} \sum_{\epsilon \in \mathcal{E}_{\text{strong}}} \frac{\Delta D'_{2, k}(\epsilon)}{\beta_2(\epsilon)}
\end{equation}
where $\mathcal{E}_{\text{strong}}$ is the discrete set of evaluated noise parameters satisfying $\epsilon > 0.25$, and $|\mathcal{E}_{\text{strong}}|$ is the cardinality of this set. 

Beyond its role as a visualization tool, the sorting score $S_k$ admits a natural operational interpretation. Since it quantifies the average proximity to the effective bound within the strongly non-Markovian regime at the two-copy level, it effectively acts as an intrinsic measure of saturation propensity for each state pair. 

Low values of $S_k$ identify configurations for which the protocol faces structural limitations already at $n=2$, whereas high values correspond to pairs that are naturally aligned with the saturation conditions of Theorem~\ref{Theorem saturating states}. In this sense, the induced ordering does not merely organize the data, but reveals a hierarchy of state pairs according to how favorably their spectral structure supports distinguishability saturation of the effective general bound under coarse graining.

The pairs are then arranged according to a permutation $\pi$ such that the scores are in ascending order:
\begin{equation}
    S_{\pi(1)} \le S_{\pi(2)} \le \dots \le S_{\pi(N)}
\end{equation}
The vertical axis in the heatmaps represents this normalized rank, $y_j = (j / N) \times 100\%$, mapping the state pair $\pi(j)$ to the $j$-th row from the bottom. 

Consequently, the $y$-axis intrinsically ranks the state pairs from those that are ``hardest'' to natively enhance at $n=2$ (bottom) to those that most easily saturate the effective general bound (top). Crucially, this exact same permutation $\pi$ is applied universally across all subsequent heatmaps---including those for $n=3$ copies, the absolute gain, and the optimal values. This rigid global ordering guarantees that any horizontal slice drawn across any panel or figure tracks the exact same physical state pair, allowing for a direct, one-to-one visual analysis of how expanding the target space to $n=3$ overcomes the limitations inherent to specific pairs at $n=2$.

In all following numerical evaluations, the freedom of the protocol is strictly encoded in the choice of the unitary operator $U$, which mediates correlations between the primary copies, the auxiliary system, and the target ancilla.

Figure~\ref{fig: tightness heatmap} illustrates the protocol's efficiency by mapping the tightness ratio, which quantifies how closely the optimized distinguishability approaches the effective general bound $\beta_n$. As the target space increases from $n=2$ to $n=3$, the tightness ratio remains remarkably stable. This consistency indicates that $\beta_n$ provides a reliable scaling benchmark for the optimal distinguishability gain, even in cases where the bound is not strictly saturated.

\begin{figure*}[ht]
    \includegraphics[width=0.75 \linewidth]{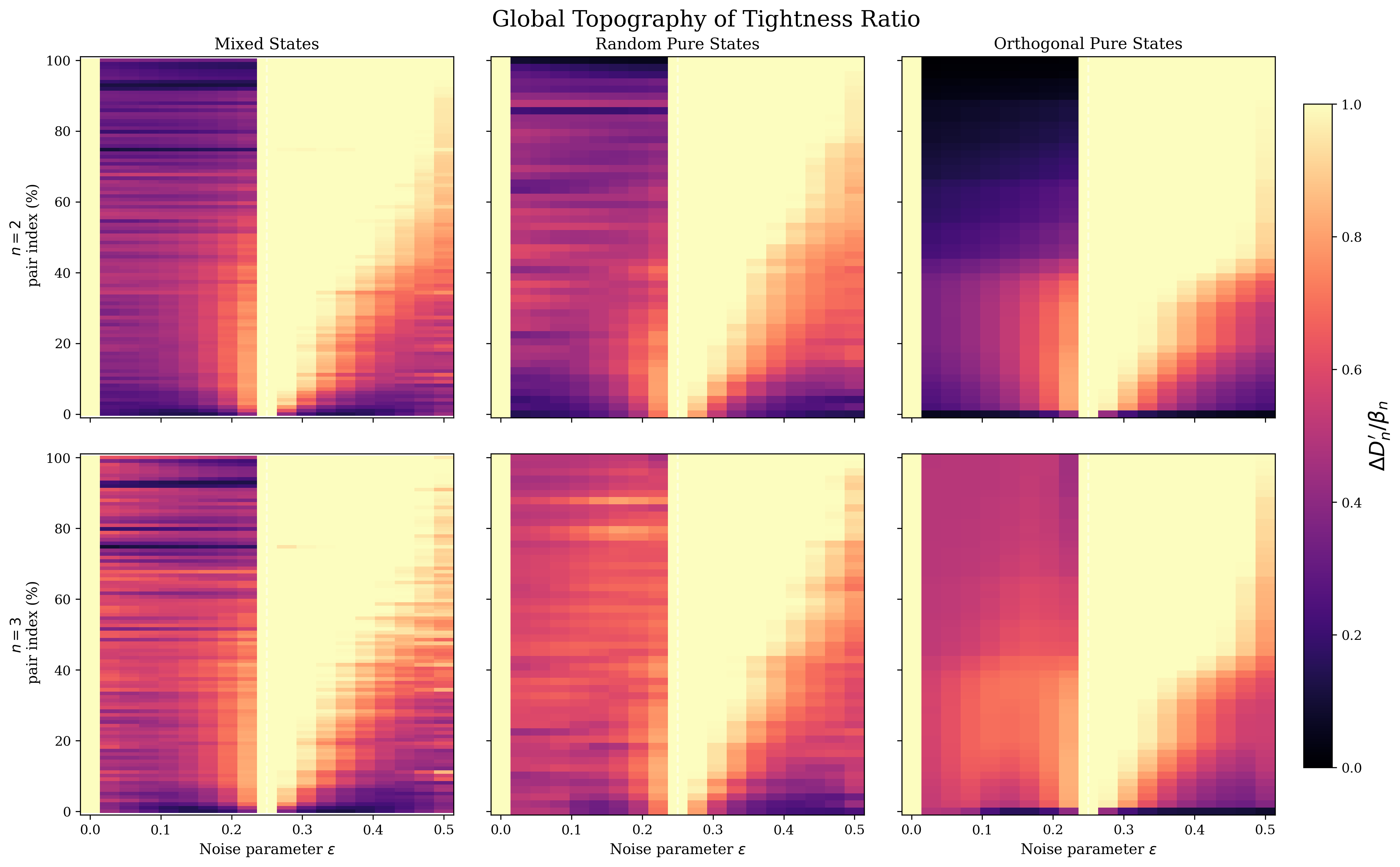}
    \caption{Global topography of the tightness ratio, $\Delta D'_n/\beta_n$, for $n=2$ (top row) and $n=3$ (bottom row) copies. The columns correspond to ensembles of mixed states, random pure states, and orthogonal pure states, respectively. This ratio quantifies the efficiency of the optimization, illustrating how closely the optimized distinguishability change approaches the effective general bound. The dashed vertical line at $\epsilon=0.25$ indicates the theoretical boundary between the weakly and strongly non-Markovian regimes. Notably, expanding the target space to $n=3$ does not result in a significant change of the ratio. Furthermore, orthogonal pure states demonstrate a distinct structural advantage in achieving maximum efficiency compared to the other ensembles, which have similar behaviors.}
    \label{fig: tightness heatmap}
\end{figure*}

Despite this global stability, the structure of the tightness landscape depends strongly on the ensemble. For orthogonal pure states, the distribution exhibits a smooth and ordered gradient, reflecting a uniform tendency to saturate the bound. In contrast, both non-orthogonal pure states and mixed states display a markedly fragmented structure, indicating a less efficient and less uniform approach to the bound.

This ensemble dependence can be understood directly from the saturation condition in Theorem~\ref{Theorem saturating states}, which is governed by the spectral structure of the operators $A$ and $B$ through the inequality
\[
\operatorname{Tr}(P_+A) + \operatorname{Tr}(E_0 A)
\;\geq\;
\operatorname{Tr}(P_+B) + \operatorname{Tr}(E_0 B)
\;\geq\; 0.
\]
For orthogonal pure states, the initial distinguishability is maximal, and the corresponding operators inherit a highly structured spectrum that facilitates the fulfillment of this condition across a broad parameter range. Although the open-system dynamics may partially suppress this distinguishability, the effective dynamics is able to recover it efficiently by exploiting multi-copy correlations. This is consistent with the general understanding that orthogonal state pairs are optimal for witnessing information backflow in non-Markovian dynamics~\cite{orthogonalstates1, orthogonalstates2}, and here we observe that they also provide a favorable structure for its recovery under coarse graining.

In contrast, for non-orthogonal pure states and mixed states, the initial distinguishability is intrinsically limited, and the spectra of the corresponding operators $A$ and $B$ become more intricate. This reduces the likelihood that the projector-based inequality can be satisfied, thereby restricting access to the saturation regime. As a consequence, the optimization must compensate through more complex redistributions of distinguishability, leading to the fragmented and less uniform patterns observed in the heatmaps. These results indicate that the efficiency of the protocol is ultimately controlled by how the spectral weight of $A$ and $B$ is distributed relative to the positive and null subspaces defined by $\Delta = A - B$, rather than by commutativity alone.

A particularly notable effect arises in the weakly non-Markovian regime, where $\Delta D \leq 0$. Even in this regime, the optimized effective dynamics yields a non-negative, and frequently positive, value of $\Delta D'_n$ (see Fig.~\ref{fig: optimization results}), thereby signaling a transition to essential non-Markovianity at the level of the effective dynamics.

The observed stability of the tightness ratio under the increase from $n=2$ to $n=3$ suggests that the effective general bound $\beta_n$ captures an intrinsic structural limitation of the protocol rather than an artifact of the optimization. In particular, while enlarging the target space increases the attainable distinguishability, it does not significantly alter how closely this optimum approaches the bound. This indicates that $\beta_n$ may encode constraints that are insensitive to the number of copies, and are instead governed by the spectral compatibility of the initial operators. Consequently, the tightness ratio should be understood as a structural indicator of how well a given state pair aligns with the saturation conditions, rather than as a direct proxy for operational advantage.

While tightness captures proximity to the bound, it does not directly quantify operational advantage. To address this, we examine the gain $\Delta D'_n - \Delta D$, shown in Fig.~\ref{fig: nM gain heatmap}. The results reveal a clear trade-off: regions exhibiting larger numerical gains tend to correspond to lower tightness ratios, particularly within the parameter space associated with strong non-Markovianity in the original dynamics. In the weakly non-Markovian regime, however, this correlation becomes less pronounced.

This behavior reflects the constraints imposed by the saturation conditions of Theorem~\ref{Theorem saturating states}. Configurations that closely approach the effective bound $\beta_n$ are structurally restricted, which limits the extent to which distinguishability can be further enhanced. Conversely, when these conditions are not met, the optimization retains greater flexibility, allowing for larger absolute gains. Thus, proximity to the bound and operational amplification probe distinct aspects of the protocol’s performance.

\begin{figure*}[h]
    \includegraphics[width=0.75\linewidth]{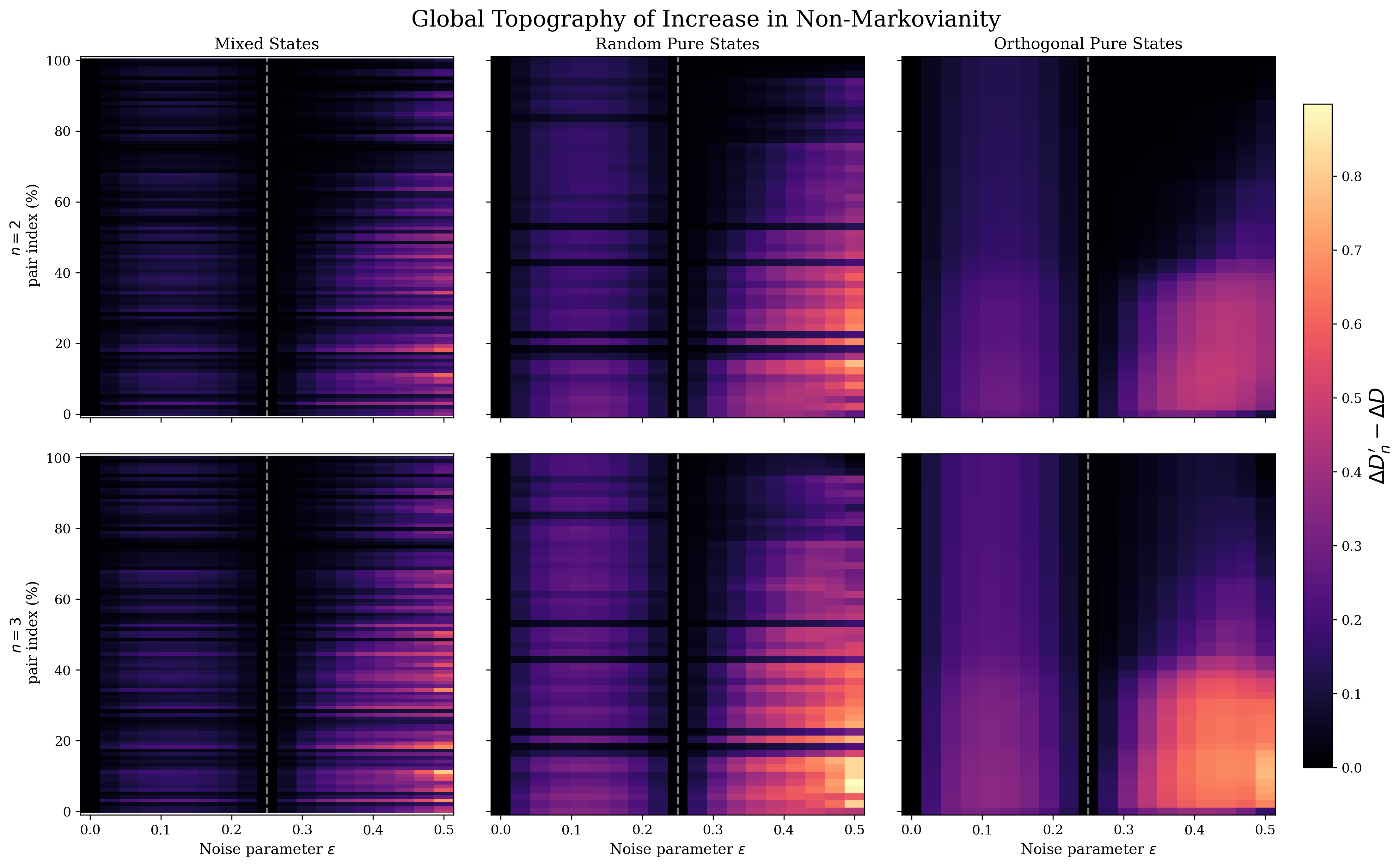}
    \caption{Global topography of the absolute increase in non-Markovianity, $\Delta D'_n - \Delta D$, for $n=2$ (top row) and $n=3$ (bottom row) copies. The columns correspond to ensembles of mixed states, random pure states, and orthogonal pure states, respectively. This difference quantifies the absolute gain achieved by the optimized effective dynamics relative to the original dynamics. The dashed vertical line at $\epsilon=0.25$ indicates the theoretical boundary between the weakly and strongly non-Markovian regimes. Expanding the target space to $n=3$ yields a substantial improvement in the gain across all ensembles, explicitly demonstrating the multi-copy advantage. Moreover, initializing the protocol with pure states yields significantly higher absolute gain compared to mixed state ensembles.}
    \label{fig: nM gain heatmap}
\end{figure*}

Across ensembles, orthogonal pure states again display a more regular and structured behavior, whereas the remaining ensembles exhibit increasingly irregular patterns. This reinforces that the interplay between $\Delta D'_n$, $\Delta D$, and $\beta_n$ is inherently ensemble-dependent. Nevertheless, a consistent trend emerges: increasing the number of copies enhances the performance of the optimized protocol in nearly all cases, highlighting the role of multi-copy processing.

Importantly, the transition from $n=2$ to $n=3$ reveals more than a quantitative improvement—it signals a qualitative threshold in the protocol's capabilities. In some cases where two copies are insufficient to induce an effective change, the addition of a third copy allows essential non-Markovianity to be witnessed. 

This behavior is indicative of a multi-copy activation mechanism, in which correlations accessible only at higher tensor powers unlock transformations that are otherwise forbidden. Such threshold effects are similar in phenomena in other areas of quantum information, including entanglement distillation and activation~\cite{Bennett1996,Horodecki1999}, as well as the superactivation of quantum channel capacities~\cite{Smith2008}.

The results also exhibit a clear dependence on the parameter $\varepsilon$. In the interval $0 \leq \varepsilon \leq 0.25$, the optimized value does not reach the bound, and the tightness ratio typically remains below $0.5$. By contrast, in the region $0.25 \leq \varepsilon \leq 0.5$, the bound is attained in a significant fraction of cases, with generally higher tightness values. However, this increase does not translate proportionally into larger gains, further underscoring that bound saturation and distinguishability amplification capture complementary aspects of the protocol’s behavior.

The distinct behavior across the two parameter intervals can be understood in terms of the underlying spectral ordering of the effective operators. In the weakly non-Markovian region, the eigenvalue structure prevents the satisfaction of the saturation conditions, enforcing low tightness and necessitating reliance on optimization-induced amplification. Conversely, in the strongly non-Markovian region, the natural ordering of eigenvalues increasingly aligns with the conditions of Theorem~\ref{Theorem saturating states}, enabling frequent saturation of $\beta_n$. However, as discussed above, this improved alignment does not translate into enhanced gain, further reinforcing the structural origin of the tightness–gain trade-off.

To better understand these effects, we turn to the extreme cases shown in Fig.~\ref{fig:regime_change_extremes}. Even in worst-case scenarios for mixed states and generic pure states, the emergence of operational signatures can, in principle, already be detected at the two-copy level through the positivity of $\Delta D'_n$. The only exception occurs for orthogonal pure states, for which this quantity remains identically zero.

\begin{figure*}[hb] 
    \centering
    \subfloat[Best-case: Mixed states\label{fig:best_mixed}]{%
        \includegraphics[width=0.3\textwidth]{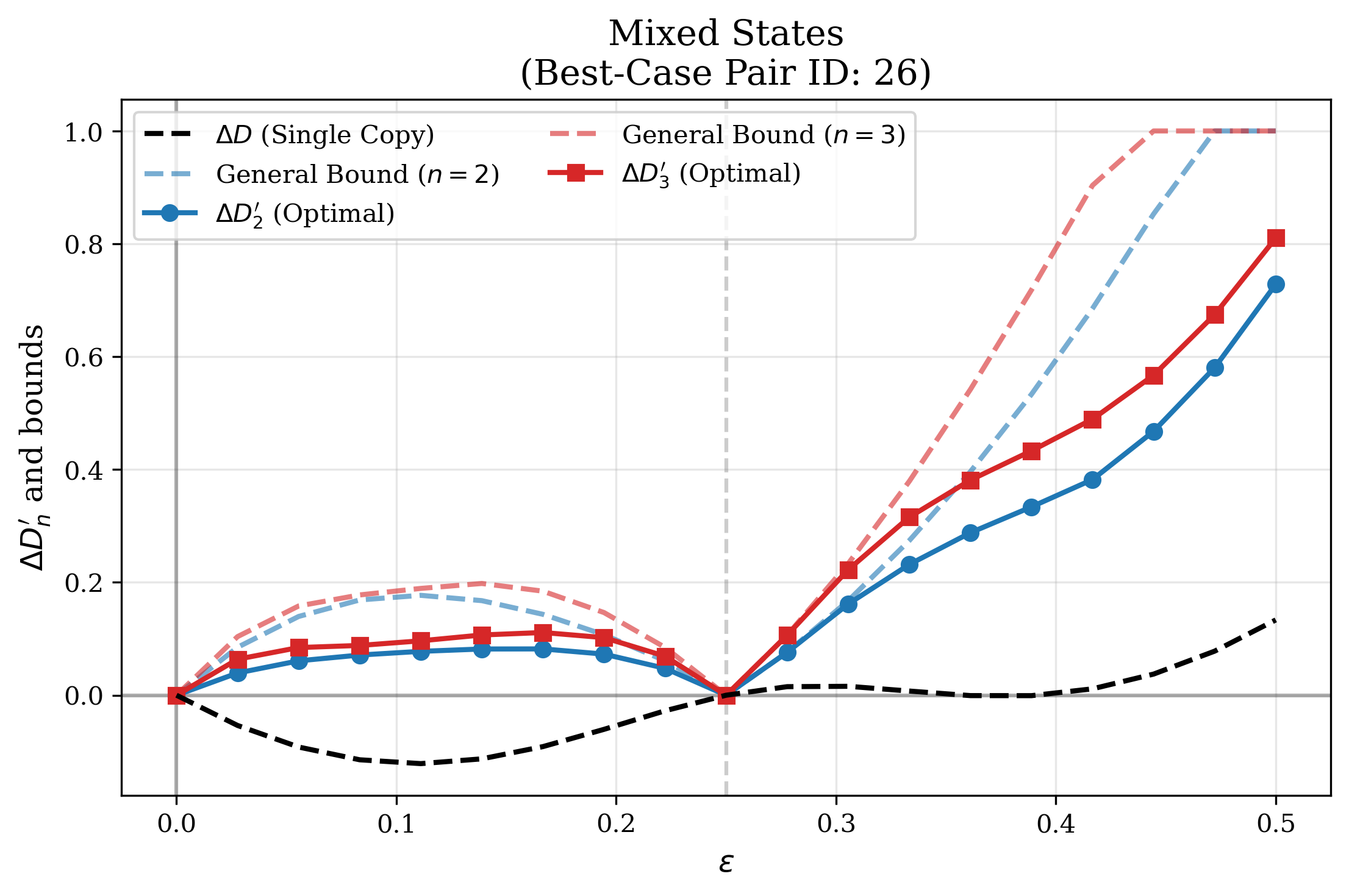}%
    }
    \hfill
    \subfloat[Best-case: Random pure states\label{fig:best_pure}]{%
        \includegraphics[width=0.3\textwidth]{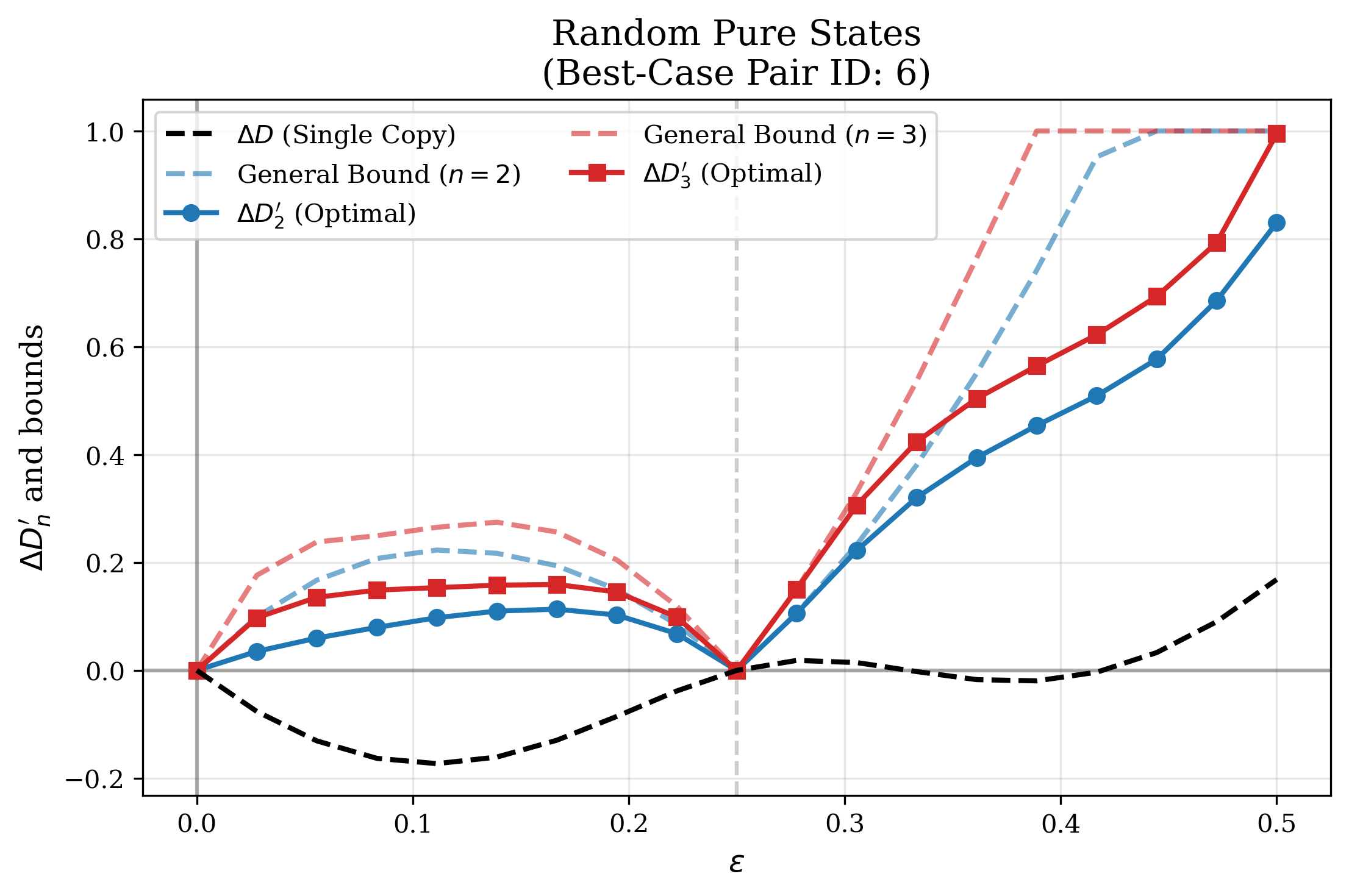}%
    }
    \hfill
    \subfloat[Best-case: Orthogonal pure states\label{fig:best_ortho}]{%
        \includegraphics[width=0.3\textwidth]{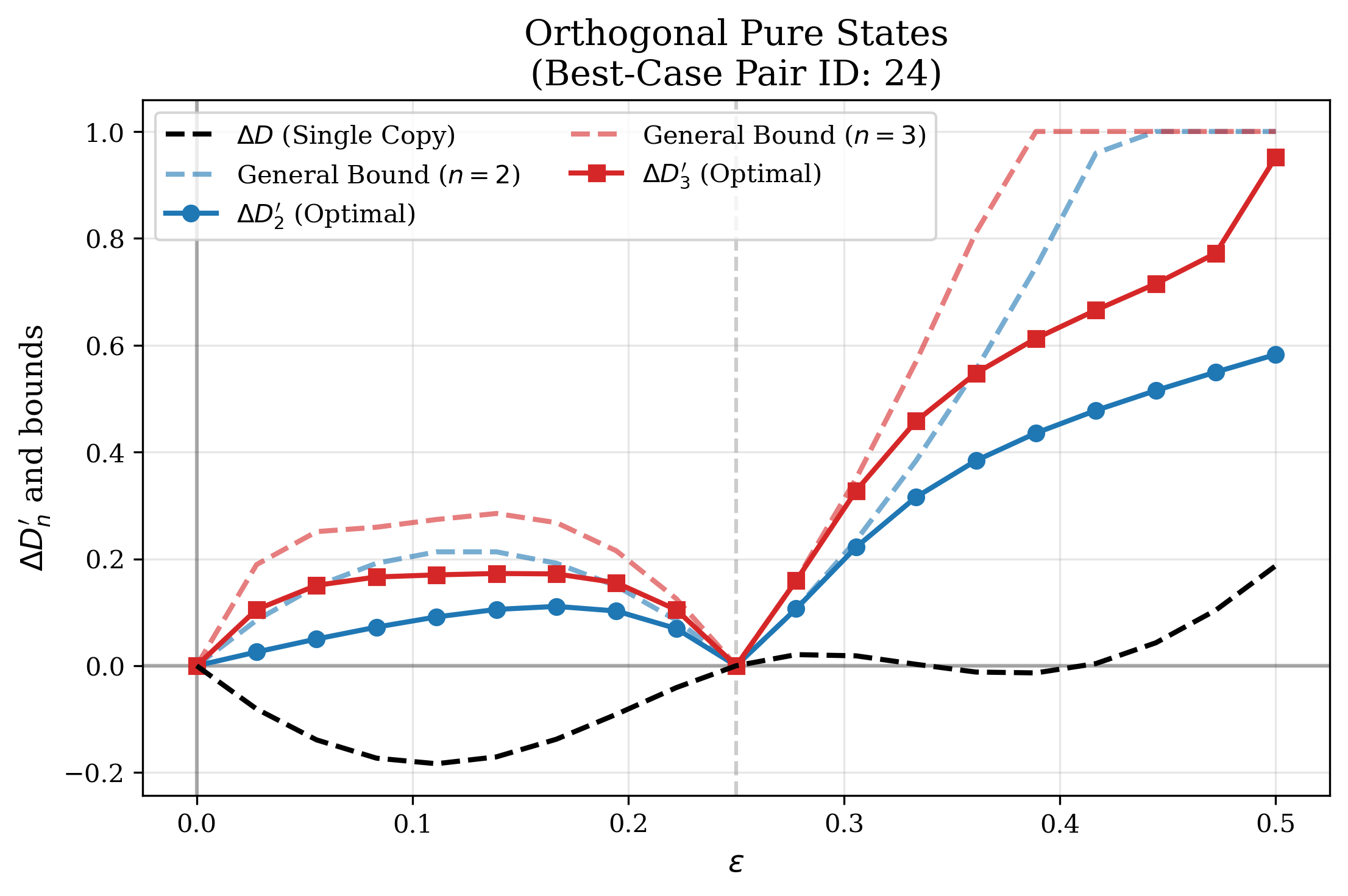}%
    }
    
    \vspace{1em}
    
    \subfloat[Worst-case: Mixed states\label{fig:worst_mixed}]{%
        \includegraphics[width=0.3\textwidth]{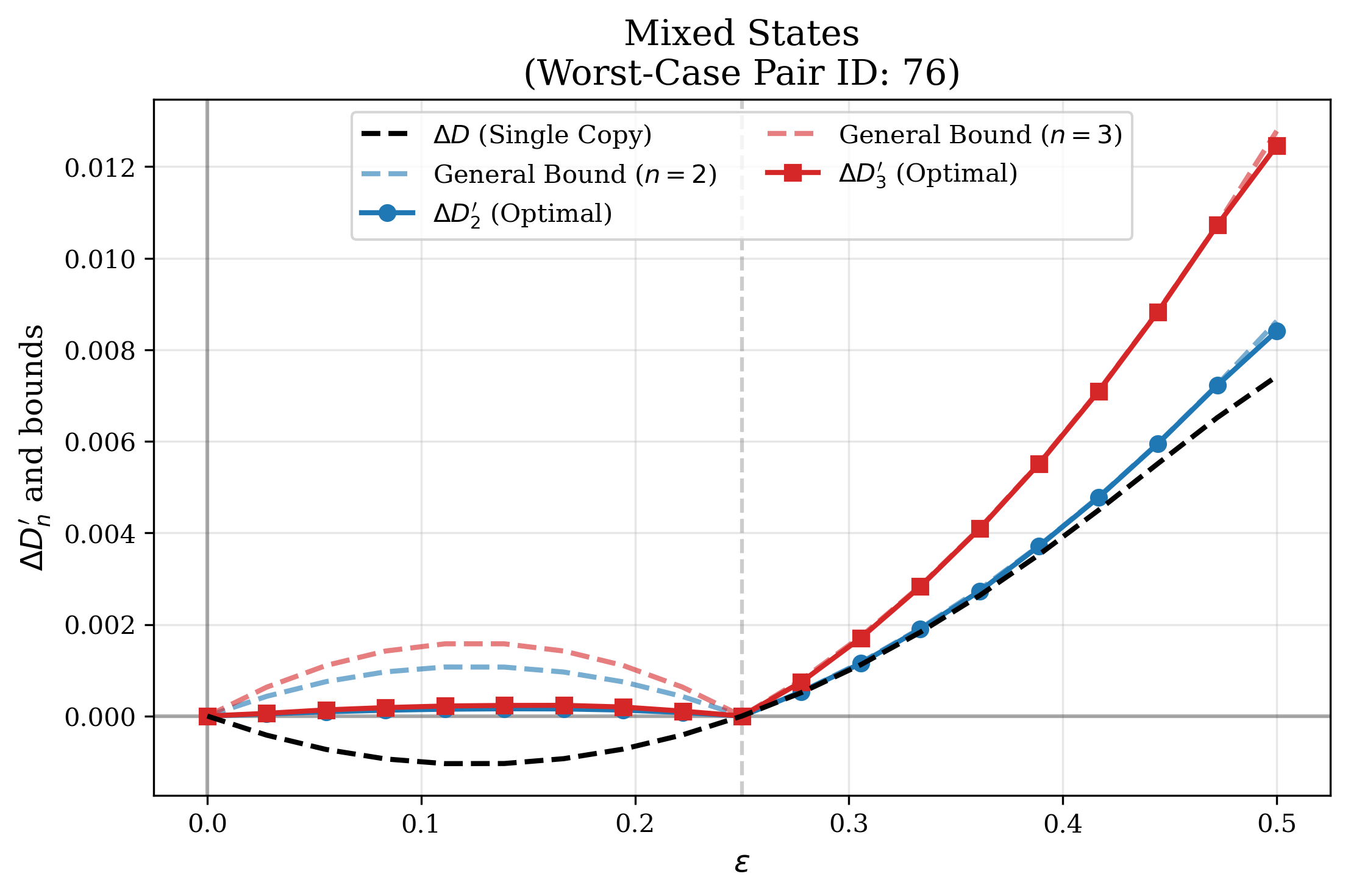}%
    }
    \hfill
    \subfloat[Worst-case: Random pure states\label{fig:worst_pure}]{%
        \includegraphics[width=0.3\textwidth]{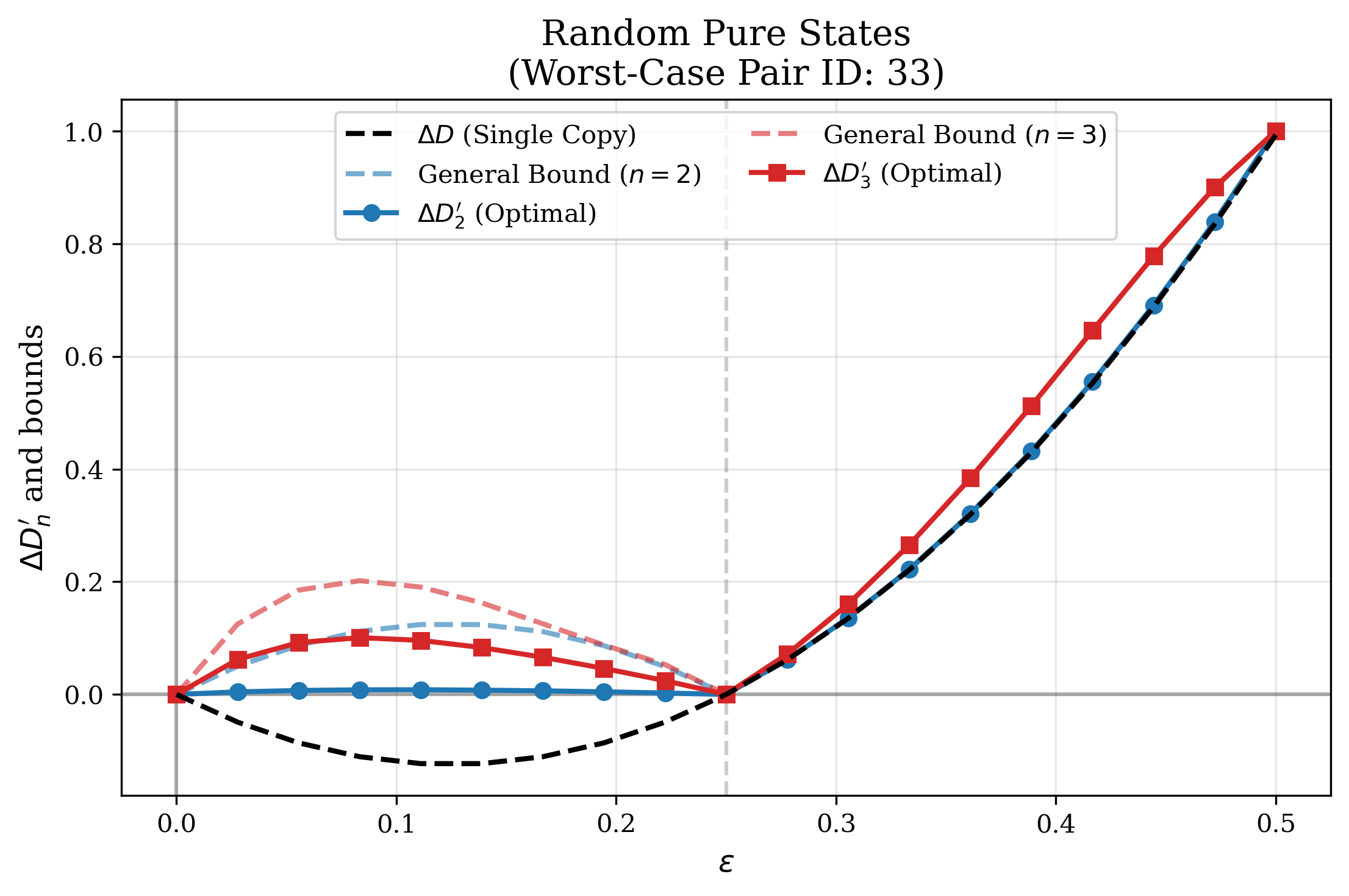}%
    }
    \hfill
    \subfloat[Worst-case: Orthogonal pure states\label{fig:worst_ortho}]{%
        \includegraphics[width=0.3\textwidth]{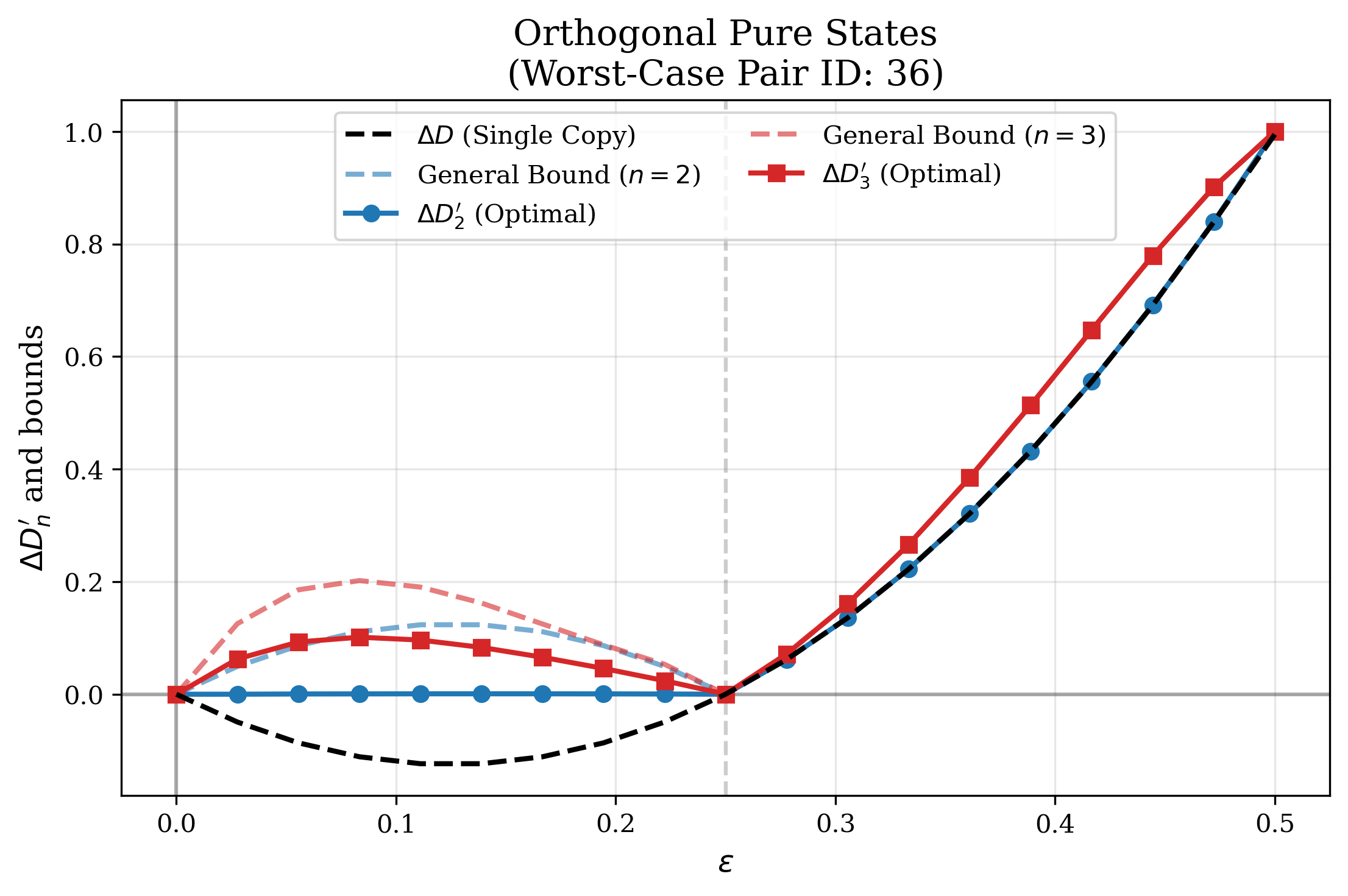}%
    }
    
    \caption{\textbf{Best- and worst-case scenarios for the effective transition to the strongly non-Markovian regime.} The top row (a--c) illustrates the best-case state pairs, representing the maximum achievable distinguishability change ($\Delta D'_n$) in the weakly non-Markovian region ($\epsilon < 0.25$). The bottom row (d--f) displays the corresponding worst-case pairs, representing the minimum positive value of $\Delta D'_n$ achieved by the protocol. The columns correspond to ensembles of mixed states (a, d), random pure states (b, e), and orthogonal pure states (c, f). These extreme scenarios delineate the absolute limits of the protocol's effectiveness: while best-case pairs readily undergo the effective transition, worst-case pairs reveal the fundamental structural bottlenecks, requiring three copies to produce an observable positive distinguishability change.}
    \label{fig:regime_change_extremes}
\end{figure*}

In practice, however, the magnitude of $\Delta D'_n$ in these extreme cases is typically very small, making experimental detection challenging. For pure states, increasing the number of copies leads to a substantial amplification of the optimized signal, potentially enabling experimental observation. In contrast, for mixed states in the worst-case regime, the values remain too small to be reliably detected, reflecting the extremely weak initial distinguishability across the parameter space.

Consistent with the trends identified above, these extreme cases also highlight that scenarios in which the bound is saturated tend to exhibit poorer overall performance, whereas larger gains are typically achieved when the bound is not reached. This behavior is further connected to the structure of the effective bound itself: the best-performing cases occur when $\beta_n$ is capped at $1$, due to the upper bound in Eq.~\eqref{ineq: upper bound protocol} exceeding unity at relatively small values of $\varepsilon$. This contrasts with the worst-case scenarios, where such saturation does not occur and performance remains limited.

Although the extreme cases were selected to emphasize the onset of the emergence, they also reveal a limitation of the protocol. In the region $0.25 \leq \varepsilon \leq 0.5$, there is essentially no improvement at the two-copy level—except for a marginal effect in the mixed-state example—and even with three copies the performance gain remains modest. This further underscores that the effectiveness of optimized multi-copy coarse-graining depends sensitively on both the parameter regime and the structure of the underlying ensemble.

\begin{figure*}[ht]
    \includegraphics[width=0.8\linewidth]{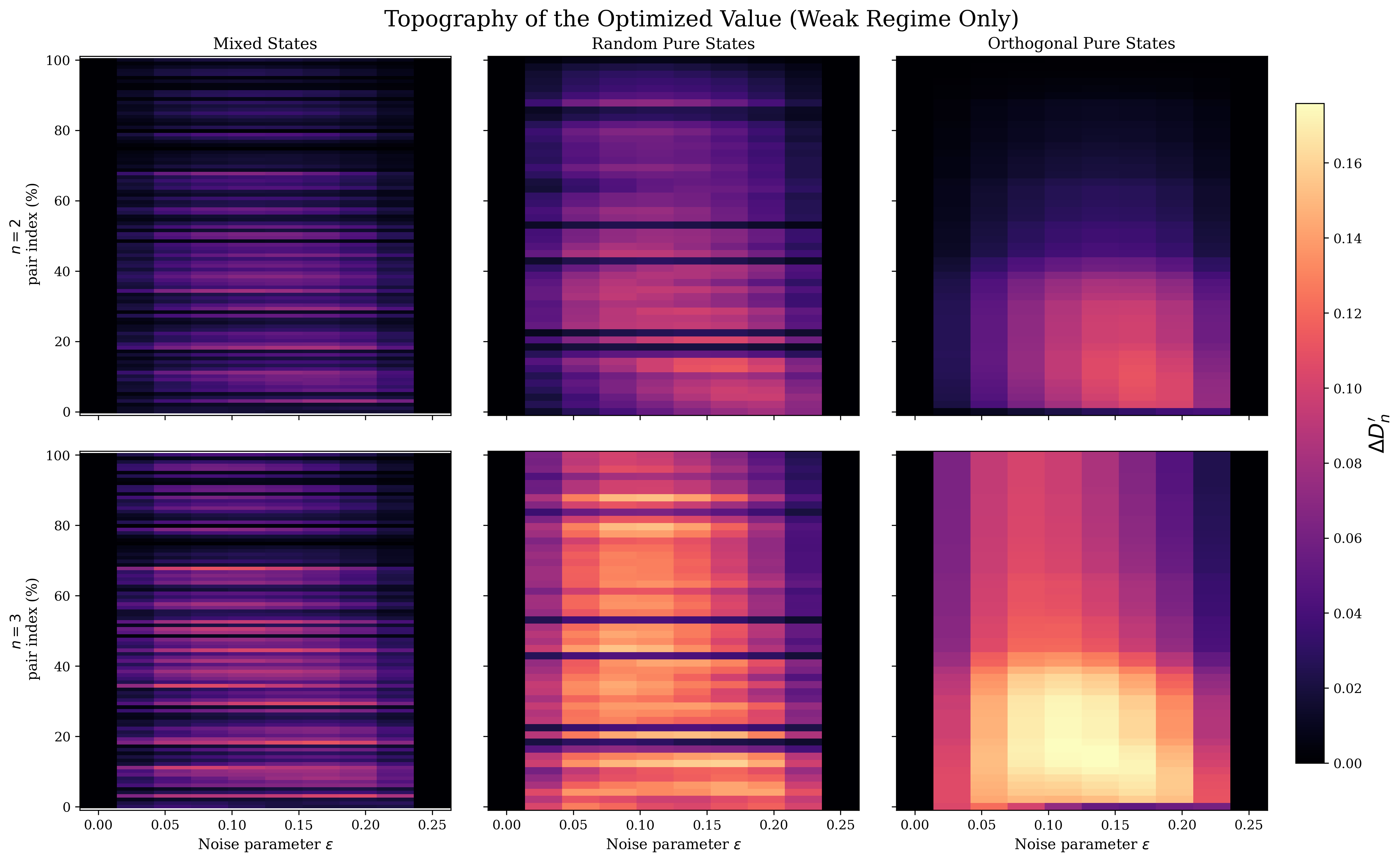}
    \caption{Heatmap of the optimized distinguishability change, $\Delta D'_n$, restricted to the parameter region $0 < \varepsilon < 0.25$, where the original dynamics is weakly non-Markovian and satisfies $\Delta D \leq 0$. The vertical axis follows the same global ordering of state pairs defined in the main text, ensuring consistency with the other figures. Despite the absence of information backflow in the original dynamics, the optimized effective dynamics yields non-negative—and in most cases strictly positive—values of $\Delta D'_n$. This behavior provides direct evidence of a protocol-induced transition to essential non-Markovianity at the level of the effective dynamics. The enhancement becomes more pronounced as the number of copies increases, illustrating the role of multi-copy processing in detecting hidden distinguishability.}
    \label{fig: optimization results}
\end{figure*}

A more detailed inspection of Fig.~\ref{fig: optimization results} reveals that the emergence of essential non-Markovianity is not restricted to a small subset of favorable state pairs, but instead occurs broadly across the entire ensemble. In particular, even state pairs that rank lowest according to the global ordering, corresponding to those that are hardest to optimize at the two-copy level, exhibit non-negative values of $\Delta D'_n$ once the target space is enlarged. This demonstrates that the observed effective transition is a systematic property of the optimized protocol rather than an isolated feature of a pair of initial states.

Importantly, this behavior complements the trade-off identified previously between tightness and operational gain. Throughout the region $\varepsilon < 0.25$, the optimized $\Delta D'_n$ never saturates the bound $\beta_n$, yet still yields a positive change in distinguishability. Interestingly, the emergence of positive distinguishability is predominantly observed in configurations where the structural conditions required for saturation are not fulfilled.

Finally, the comparison between different copy numbers shows that multi-copy processing plays a dual role: it not only enhances the magnitude of the change of distinguishability, but also allows a broader range of initial states to achieve positivity at the effective dynamics level. Indicating that access to higher tensor powers is a reliable ingredient to improve optimization efficiency.

\section{Discussion}
\label{sec: discussion}

Despite the underlying dynamics not being directly transformed by the protocol, the results demonstrate that an effective dynamics with a stronger operational memory signature can be constructed through collective processing. One possible scenario where such an approach could prove advantageous is when the microscopic system-environment interaction is fixed by the available hardware, making direct engineering of the dynamics impractical, thus providing a possible alternative operational route when reservoir engineering is challenging.

This flexibility does not come for free. It is necessary to have independent realizations of the same dynamics, prepare the initial states, while having to implement and maintain the CPTP map that is applied to the multiple copies. 

Whether the exchange of control over the microscopic dynamics compensates the need for a properly tuned post-processing and copies of the microscopic dynamics to achieve the enhancement of non-Markovianity will ultimately depend on the target information-processing task.

Furthermore, the results demonstrate that, despite the distinction between weak and strong non-Markovianity at the microscopic level, where no one-qubit CPTP map can convert a weak non-Markovian channel to a strong one, collective processing through coarse-graining can produce an effective dynamics with distinct observable informational properties by integrating out degrees of freedom. This fits within a broader paradigm in physics whereby coarse-graining produces effective descriptions whose operational properties differ qualitatively from those of the underlying microscopic dynamics.

This phenomenon is well illustrated by effective attractive interactions in superconductivity \cite{SuperconductivityBCS}, after integrating out phonon degrees of freedom, the increase in resistance at low temperatures in Kondo effect\cite{Kondoeffect,RGKondo} and Nakajima–Zwanzig projection generating memory kernels \cite{firstNakajima,secondNakajimaZwanzig,BreuerBOOK}.  In this perspective, the present protocol should not be interpreted as modifying the microscopic open-system evolution, but rather as constructing an effective dynamical description with distinct operational memory properties.

Accordingly, the results show that coarse-graining can qualitatively alter the memory properties assigned to a process. The observed information backflow is not necessarily a direct amplification of memory already present in the original dynamics, but may emerge from the effective description induced by collective processing.

The effective dynamics is defined by the completely positive maps available to the observer. If an observer has access to $n$ copies and can apply a joint CPTP map, the effective dynamics $\Lambda_t'$ is the physically relevant description for their information processing tasks. The memory is not just in the environment anymore; it is encoded in the cross-correlations of the 
$n$ copies, which the coarse-graining map extracts.

Interestingly, a trade-off was observed between mathematical optimality (bound saturation) and operational optimality (enhancement of distinguishability change). This suggests that satisfying the constraint imposed by Theorem\ref{Theorem saturating states} imposes strong limitations on the enhancement of distinguishability change.  

A central outcome of our analysis is the systematic generation of essential non-Markovianity from weakly non-Markovian dynamics. In almost all cases considered, and across the sampled range of $\varepsilon$, the optimized protocol yields a non-negative, and typically strictly positive, effective change in distinguishability. Since the violation of contractivity for a single pair of states suffices to certify the non-positivity of the intermediate map, this demonstrates that the effective dynamics $\Lambda'_t$ becomes essentially non-Markovian when optimized.

More precisely, the optimal unitary $U^\ast$ defines a coarse-graining map $\lambda_{U^\ast}$ which induces an effective dynamics $\Lambda'_t$~\eqref{effective channel} that exhibits information backflow, even when the original dynamics does not. Notably, this effective transition can already occur at the two-copy level when the unitary is suitably chosen. However, its observability depends on the structure of the input states: for certain pairs—such as those saturating the general bound—additional copies may be required to reveal the transition through the optimization of $\Delta D'_n$.

The optimization results exhibit a strong form of robustness. Because the coarse-graining map in Eq.~\eqref{coarse-graining} depends continuously on the unitary $U$, and the trace norm is itself continuous, the quantity $\Delta D'_n(\rho_1,\rho_2,U)$ varies continuously with respect to both the unitary and the input states. Consequently, if there exists a unitary $U^\ast$ and a pair of states $(\rho_1,\rho_2)$ such that $\Delta D'_n > 0$, then there exists an open neighborhood of unitaries and state pairs for which this positivity condition is preserved. 

This implies that the generation of essential non-Markovianity does not rely on fine-tuned control: near-optimal unitaries and approximate state preparations are sufficient to induce the effective transition. In particular, the continuity with respect to the input states ensures robustness against imperfections in state preparation, while continuity in the unitary guarantees stability under experimental control errors. Since the optimization yields, in general, different unitaries for different values of $\varepsilon$, the full set of activating transformations is given by the union of these neighborhoods, forming a nontrivial region in parameter space where strong non-Markovianity can be reliably observed.

The present work raises a broader question beyond non-Markovianity itself. To what extent do effective descriptions obtained through coarse-graining preserve the operational properties of the underlying microscopic dynamics? The generation of effective essential non-Markovianity demonstrated here suggests that such properties need not be invariant under changes of description. It would therefore be interesting to investigate whether analogous phenomena occur in many-body systems, where coarse-graining is a central ingredient in the construction of effective theories.

\section{Conclusions}

We have shown that the optimized effective dynamics provides a systematic and quantitatively controlled method for enhancing distinguishability in open quantum dynamics. At the quantitative level, the protocol achieves substantial gains relative to the original evolution. More importantly, it enables a qualitative transformation: weakly non-Markovian dynamics—characterized by the absence of observable information backflow—can be promoted to essentially non-Markovian dynamics, where backflow becomes operationally detectable through a positive change in trace distance.

From an operational perspective, this demonstrates that the observable degree of non-Markovianity is not determined solely by the microscopic dynamics, but also by the class of admissible processing operations used to interrogate it. Collective processing, therefore, provides an alternative route for enhancing operational memory signatures when direct control over the system-environment interaction is unavailable or impractical.

Beyond the effective transition itself, the work establishes a general optimization framework for distinguishability recovery, including a computable upper bound, necessary and sufficient saturation conditions, and a quantitative measure of protocol performance. The observed trade-off between mathematical optimality (bound saturation) and operational optimality further shows that saturating the theoretical bound does not necessarily maximize the attainable distinguishability gain.

The protocol is also shown to be operationally robust. The effective transition already appears at the two-copy level, becomes more systematic with additional copies, and remains stable under small enough perturbations of both the optimizing unitary and the input states owing to continuity of the coarse-graining map and the trace norm. These features suggest that the protocol may tolerate realistic imperfections in experimental implementations.

More broadly, the present work suggests that informational properties of quantum dynamics need not be invariant under the construction of effective descriptions. Understanding which dynamical features are preserved under coarse-graining, and which may emerge only at the effective level, may provide a broader perspective on the role of memory in open quantum systems and on the characterization of effective quantum dynamics.

\section*{Acknowledgements}

G. Moniz acknowledges financial support from the  Fundação de Amparo à Pesquisa do Estado de São Paulo (FAPESP), BCO Doutorado Direto, through Grant No. 2024/09165-7.

B. Amaral acknowledges financial support from the Instituto Serrapilheira, Chamada No. 4 2020 and Fundação de Amparo à Pesquisa do Estado de São Paulo, Auxílio à Pesquisa—Jovem Pesquisador, through Grant No. 2020/06454-7 (until September 2025). 

We further acknowledge the use of \emph{Gemini} to assist in the proof of Theorem \ref{Theorem saturating states}.


\bibliographystyle{unsrt}
\bibliography{references.bib}

\appendix

\section{Contractivity of trace-1 norm for positive maps and its connection to non-Markovianity regimes}
\label{app: distinguishability as witness of the Regime}

\begin{lemma}
\label{lemma: positivity}
    For every linear positive trace preserving map $\Phi : \mathcal{B}(H) \rightarrow \mathcal B(H)$, it holds
    \begin{equation}
        ||\Phi(A)||_1 \leq ||A||_1
    \end{equation}
    for every operator $A \in \mathcal{B}(H)$.
\end{lemma}

\begin{proof}
    Consider $Q \in \mathcal{B}(H)$ hermitian. Thus, using Jordan decomposition, we can write $Q= Q_+ -Q_-$, where $Q_+,\ Q_-$ are positive operators, with orthogonal support. It follows two facts:
    \begin{itemize}
        \item $||Q||_1 = \Tr(Q_+) + \Tr (Q_-); \qquad \Tr(Q) =\Tr(Q_+) - \Tr (Q_-).$
    \item $\Phi(A_{\pm})$ is positive, thus $||\Phi(Q_{\pm})||=\Tr(\Phi(Q_{\pm}))=\Tr(Q_{\pm})$
    \end{itemize}
    
    Now, we simply use the triangular inequality and linearity of $\Phi$
    \begin{equation}
        ||\Phi(Q)||_1 = ||\Phi(Q_+)-\Phi(Q_-)||_1\leq ||\Phi(Q_+)||_1 + ||\Phi(Q_-)||_1 = ||Q||_1,
    \end{equation}
    which extends to anti-hermitian matrices via the isomorphism $Q\rightarrow iQ$.
    
    Finally, any operator $A \in \mathcal{B}(H)$ can be decomposed as a sum of its hermitian and anti-hermitian components $(A+i A^\dagger)/2$ and $(A-i A^\dagger)/2$, so again via the triangular inequality and linearity of $\Phi$
    \begin{equation*}
        ||\Phi(A)||_1 \leq ||A||_1
    \end{equation*}
    as desired.
\end{proof}

\begin{lemma}
Let $\{\Lambda_1, \Lambda_2\}$ be a time-discrete quantum evolution such that
\begin{equation}
    \Lambda_2 = V_{2,1} \circ \Lambda_1,
\end{equation}
where the intermediate map $V_{2,1}$ is positive and trace-preserving. Then, for any operator $A \in \mathcal{B}(\mathcal{H})$, it holds that
\begin{equation}
    \|\Lambda_2(A)\|_1 \leq \|\Lambda_1(A)\|_1.
\end{equation}
\end{lemma}

\begin{proof}
Using the decomposition $\Lambda_2 = V_{2,1} \circ \Lambda_1$ and applying Lemma~\ref{lemma: positivity}, which states that positive trace-preserving maps are contractive with respect to the trace norm, we obtain
\begin{equation}
    \|\Lambda_2(A)\|_1 = \|V_{2,1}(\Lambda_1(A))\|_1 \leq \|\Lambda_1(A)\|_1.
\end{equation}
\end{proof}

\section{Saturation of the triangular inequality}
\label{app: Conditions Lemma}

\renewcommand{\thesection}{\Alph{section}}
\begin{lemma}[Norm duality {\cite[Eq.\ 1.173]{Watrous_2018}}]
For any $A \in M_n(\mathbb{C})$ we have
\begin{equation}
    \|A\|_1 = \max \{ \operatorname{Tr}(BA) : B \in M_n(\mathbb{C}),\ \|B\|_\infty \leq 1 \}.
\end{equation}
\end{lemma}

\begin{remark}
If $A$ is Hermitian, then the maximizer $B$ must be Hermitian.  
Indeed, otherwise $\operatorname{Tr}(BA)$ could be non-real, which contradicts the fact that the trace norm is real-valued. In fact, a maximizer is the \emph{sign operator} of $A$, i.e.\ $B = \mathrm{sign}(A)$.
\end{remark}

\begin{theorem}
    
Let $A$ and $B$ be $n \times n$ Hermitian matrices. Then
\[
 \|A\|_1 + \|B\|_1 = \|A+B\|_1
\]
if and only if $A$ and $B$ commute and, for every common eigenvector $|v\rangle$ with 
\[
A|v\rangle = v_a |v\rangle, \qquad B|v\rangle = v_b |v\rangle,
\]
the corresponding eigenvalues satisfy $v_a v_b \geq 0$.
\label{theorem triangular}
\end{theorem}

\begin{proof}
Suppose first that $A$ and $B$ commute. Then they are simultaneously diagonalizable: there exists a unitary $U$ such that
\[
A = U \,\mathrm{diag}(a_1,\dots,a_n)\, U^\dagger, 
\qquad 
B = U \,\mathrm{diag}(b_1,\dots,b_n)\, U^\dagger.
\]
Hence
\[
A+B = U \,\mathrm{diag}(a_1+b_1,\dots,a_n+b_n)\, U^\dagger,
\]
and therefore
\[
\|A+B\|_1 = \sum_{i=1}^n |a_i+b_i|.
\]
On the other hand,
\[
\|A\|_1 + \|B\|_1 = \sum_{i=1}^n |a_i| + \sum_{i=1}^n |b_i|.
\]
These two expressions are equal if and only if $a_i b_i \geq 0$ for all $i$.

Conversely, assume $\|A+B\|_1 = \|A\|_1 + \|B\|_1$.  
By Lemma~A.1 there exist matrices $U,V,W$ with $\|U\|_\infty,\|V\|_\infty,\|W\|_\infty \le 1$ such that
\begin{equation*}
\begin{split}
    &\|A\|_1 = \operatorname{Tr}(UA), \qquad
\|B\|_1 = \operatorname{Tr}(VB),\\
&\|A+B\|_1 = \operatorname{Tr}(W(A+B)).
\end{split}
\end{equation*}

Since
\[
\operatorname{Tr}(W(A+B)) = \operatorname{Tr}(WA) + \operatorname{Tr}(WB),
\]
and each term is bounded above by the corresponding norm, equality forces
\[
\operatorname{Tr}(WA) = \|A\|_1, \qquad \operatorname{Tr}(WB) = \|B\|_1.
\]
Thus the same maximizer $W$ works simultaneously for $A$ and $B$.  

Now let $\{|a_i\rangle\}$ be an eigenbasis of $A$ with $A|a_i\rangle = a_i |a_i\rangle$. Then
\[
\|A\|_1 = \sum_i a_i \langle a_i |W|a_i\rangle \le \sum_i |a_i|,
\]
with equality only if $\langle a_i|W|a_i\rangle = \operatorname{sign}(a_i)$ when $a_i \neq0$.  
But since $\|W\|_\infty \le 1$ and $W$ is Hermitian, the condition 
\(\langle a_i|W|a_i\rangle = \pm 1\) implies $W|a_i\rangle = \pm |a_i\rangle$, since
\begin{equation}
    \begin{split}
        || (W\mp I)|a_i\rangle||^2 &= \langle a_i | (W \mp I)^2|a_i\rangle \\
    &= 1 \mp 2 \langle a_i | W |a_i\rangle  + \langle a_i | W^2 |a_i\rangle\\
    &= -1 + \langle a_i | W^2 |a_i\rangle \\
    &\leq 0,
    \end{split}
\end{equation}  
hence each eigenvector of $A$ with nonzero eigenvalue is also an eigenvector of $W$. The same argument applies to $B$, so every eigenvector of $B$, with non-zero eigenvalue, is also an eigenvector of $W$. 

Let $\mathcal S_A$ (resp.\ $\mathcal S_B$) denote the span of eigenvectors of $A$ (resp.\ $B$) corresponding to nonzero eigenvalues, and set 
\[ \mathcal S:=\mathcal S_A\cup\mathcal S_B. \]

 On the subspace $\mathcal S$, $A$ and $B$ are simultaneously diagonalizable since they share the $W$-eigenbasis there. To handle the case where a element of $S$ belongs to $\ker A$, let 
 \[ X := \{\, |a_i\rangle: a_i = 0\ \text{and}\ |a_i\rangle \in \mathcal S_B \,\} \subset S_B-S_A \]
 be the elements of the eigenbasis of $A$ with eigenvalue $0$ that lie in $S_B$. Let $k = |X|$ be the number of such vectors. Then there exist $k$ eigenvectors $|b_i\rangle$ of $B$ spanning the same subspace as $X$. Each $|b_i\rangle$ can be expressed as a linear combination of the vectors in $X$, so they are also eigenvectors of $A$ with eigenvalue $0$.
 
 The same argument applies symmetrically for $B$ and in both cases the product of the associated eigenvalues for a commum eigenvector is 0.
 
 Vectors in the orthogonal complement $\mathcal S^\perp = \ker A \cap \ker B$ are unconstrained by $A$ and $B$, and any linear combination of them is an eigenvector with eigenvalue $0$ for both matrices.
 
 Therefore, we can extend the $W$-eigenbasis in $\mathcal S$ with these vectors in the zero-eigenvalue subspace to obtain a full orthonormal basis of $\mathbb C^n$ that simultaneously diagonalizes $A$ and $B$. Hence $[A,B] = 0$.

\end{proof}

\begin{corollary}
Let $A, B \in \mathcal{B}(H_m)$ be traceless Hermitian matrices. Let $\Delta = A - B$ and $P_+$ be the projector onto the strictly positive eigenspace of $\Delta$.
\[
 \|A\|_1 - \|B\|_1 = \|A-B\|_1
\]
if and only if $A$ and $B$ commute and, for every common eigenvector $|v\rangle$ with 
\[
A|v\rangle = v_a |v\rangle, \qquad B|v\rangle = v_b |v\rangle,
\]
the corresponding eigenvalues satisfy $\text{Tr}(P_+B) \geq 0$.
\label{corollary cond 1}
\end{corollary}

\begin{proof}
Set $B=W$ and $A=R+W$. The condition becomes
\[
\|R+W\|_1=\|R\|_1+\|W\|_1.
\]
By Theorem \ref{theorem triangular}, this holds iff $[R,W]=0$ and $v_r v_w\ge0$ for eigenvalues corresponding to the same eigenvector.

Rewriting in terms of $A$ and $B$, this gives $[A,B]=0$ and $(v_a-v_b)v_b\ge0$. The latter condition can be expressed as $\mathrm{Tr}(P_+B)\ge0$. Indeed, for eigenvectors of $\Delta$ with positive eigenvalues ($v_a-v_b>0$), the inequality implies $v_b\ge0$, so their contribution to $\mathrm{Tr}(P_+B)$ is non-negative.

For eigenvectors with $v_a-v_b\le0$, the inequality requires $v_b\le0$. Since both $A$ and $B$ are traceless,
\[
\mathrm{Tr}(P_+B)=-\mathrm{Tr}((I-P_+)B),
\]
so this condition is automatically satisfied once $\mathrm{Tr}(P_+B)\ge0$.
\end{proof}

\begin{lemma}
Given $\Delta \in \mathcal B(H_m)$ a Hermitian matrix and a CPTP map $\lambda:  \mathcal B(H_m) \rightarrow  \mathcal B(H_d)$. Take the spectral decompositions $\Delta= \Delta_+ - \Delta_-$ and $\lambda(\Delta) = Q_+ - Q_-$. Let $\Pi_{+}$ and $\Pi_{-}$ be the projectors onto the supports of $Q_+$ and $Q_-$, respectively. We have 
\begin{equation*}
    \|\lambda_U(\Delta)\|_1 = \|\Delta\|_1
\end{equation*}
if and only if $\Pi_{\pm}\lambda(\Delta_{\pm})\Pi_{\pm}=\lambda(\Delta_{\pm})$ and $\Pi_{\pm}\lambda(\Delta_{\mp})\Pi_{\pm}=0$.
\label{lemma cond 2}
\end{lemma}

\begin{proof}
     We have that $Q_+ - Q_- =\lambda(\Delta_+)-\lambda(\Delta_-)$ , projecting the equation using $\Pi_{\pm}$ leads to
    \begin{equation*}
    Q_{\pm} = \Pi_{\pm}\lambda(\Delta_{\pm})\Pi_{\pm} -\Pi_{\pm}\lambda(\Delta_{\mp})\Pi_{\pm}.
    \end{equation*}
    Taking the trace yields 
    \begin{equation*}
        \text{Tr}(Q_{\pm}) \leq \text{Tr}(\lambda(\Delta_{\pm}))
    \end{equation*}
    with equality, iff $\Pi_{\pm}\lambda(\Delta_{\pm})\Pi_{\pm}=\lambda(\Delta_{\pm})$ and $\Pi_{\pm}\lambda(\Delta_{\mp})\Pi_{\pm}=0$. Since equality is also equivalent to  $||\lambda(\Delta)||_1=||\Delta||_1$, the statement follows.
\end{proof}

\section{General Bound Saturation Analysis}
\label{app: general bound theorem}

\begin{theorem}
    Let $A, B \in \mathcal{B}(H_m)$ be traceless Hermitian matrices, and let $\Delta = A - B$. Let $P_+$ (resp. $P_-$) be the projector onto the strictly positive (resp. negative) eigenspace of $\Delta$ and $P_0=I-P_+-P_-$ on its Kernel. There exists a CPTP-map $\lambda: \mathcal{B}(H_m) \to \mathcal{B}(H_d)$, for $d\geq 2$, such that
    \begin{equation}
        \|\lambda(A)\|_1-\|\lambda(B)\|_1=\|\lambda(A-B)\|_1=\|A-B\|_1
    \end{equation}
    if and only if there is $0\leq E_0\leq P_0$ such that 
    $$\operatorname{Tr}(P_+A)+\operatorname{Tr}(E_0A)\geq\operatorname{Tr}(P_+B)+\operatorname{Tr}(E_0B)\ge0.$$ 
\end{theorem}

\begin{proof}
    Let $\Delta = \Delta_+ - \Delta_-$ be the spectral decomposition of $\Delta$. The projector onto the support of $\Delta_+$ is $P_+$. 

    \textbf{Sufficiency ($\impliedby$):} 
    Assume there is $0\leq E_0\leq P_0$ such that 
    $$\operatorname{Tr}( (P_++E_0) A) \geq \operatorname{Tr}( (P_++E_0) B)\ge 0.$$ We explicitly construct a maximal coarse-graining map $\lambda_{cg}: \mathcal{B}(H_m) \to \mathcal{B}(H_d)$, defined as a measurement onto the orthogonal basis elements $\{|0\rangle, |1\rangle\}$ of $H_d$ (on the basis $\{|k\rangle\}_{k=0}^{d-1}$ for $H_d$):
    \begin{equation}
        \lambda_{cg}(X)=\text{Tr}((P_++E_0)X)|0\rangle\langle0|+\text{Tr}((I-P_+-E_0)X)|1\rangle\langle1|.
    \end{equation}
    Because $0 \le E_0 \le P_0$, the operators $E_0$ and $P_0 - E_0$ are both positive. Consequently, $\{P_+ + E_0, I - P_+ - E_0\}$ constitutes a valid Positive Operator-Valued Measure (POVM) that sums to the identity. Therefore, $\lambda_{cg}$ is a well-defined measurement channel and is strictly CPTP.

    Since $E_0$ is supported on the Kernel, $E_0\Delta=0$. Then, applying this map to $\Delta$, we obtain:
    \begin{equation}
        \lambda_{cg}(\Delta)=\text{Tr}(\Delta_+)|0\rangle\langle0|-\text{Tr}(\Delta_-)|1\rangle\langle1|=\|\Delta_+\|_1|0\rangle\langle0|-\|\Delta_-\|_1|1\rangle\langle1|.
    \end{equation}
    Because $|0\rangle\langle0|$ and $|1\rangle\langle1|$ are orthogonal, we have $\|\lambda_{cg}(\Delta)\|_1 = \|\Delta_+\|_1 + \|\Delta_-\|_1 = \|\Delta\|_1$. Thus, the second equality of Eq.~\eqref{eq:saturation_condition} is satisfied.

    Next, we apply $\lambda_{cg}$ to $B$ and to $A = \Delta + B$, which by linearity is $\lambda_{cg}(A) = \lambda_{cg}(\Delta) + \lambda_{cg}(B)$. The outputs are diagonal in the $\{|k\rangle\}$ basis, so their trace norms are simply the sum of the absolute values of their coefficients:
    \begin{align}
        \|\lambda_{cg}(B)\|_1 &= |\text{Tr}((P_++E_0)B)|+|\text{Tr}((I-P_+-E_0)B)|, \\
        \|\lambda_{cg}(A)\|_1 &= \big|\|\Delta_+\|_1+\text{Tr}((P_++E_0)B)\big|+\big|-\|\Delta_-\|_1+\text{Tr}((I-P_+-E_0)B)\big|.
    \end{align}
    By our assumption, $\text{Tr}((P_++E_0)B)\ge0$ and so $\text{Tr}((I-P_+-E_0)B)=-\text{Tr}((P_++E_0)B)\le0$, due to $B$ being traceless. Since the norms are non-negative ($\|\Delta_+\|_1, \|\Delta_-\|_1 \ge 0$), we can drop the absolute values appropriately:
    \begin{align}
        \|\lambda_{cg}(B)\|_1 &= \text{Tr}((P_++E_0)B)-\text{Tr}((I-P_+-E_0)B), \\
        \|\lambda_{cg}(A)\|_1 &= \|\Delta_+\|_1+\text{Tr}((P_++E_0)B)+\|\Delta_-\|_1-\text{Tr}((I-P_+-E_0)B).
    \end{align}
    Subtracting the two yields:
    \begin{equation}
        \|\lambda_{cg}(A)\|_1-\|\lambda_{cg}(B)\|_1=\|\Delta_+\|_1+\|\Delta_-\|_1=\|\Delta\|_1.
    \end{equation}

    Notice that $\operatorname{Tr}((P_++E_0)A)\geq \text{Tr}((P_++E_0)B)$ is required for consistency, since $$\text{Tr}((P_++E_0)A)=||\Delta_+||_1 +\text{Tr}((P_++E_0)B)$$

    \vspace{0.5cm}
    \textbf{Necessity ($\implies$):}
    Assume there exists an arbitrary CPTP map $\lambda$ satisfying Eq.~\eqref{eq:saturation_condition}. From Lemma \ref{lemma cond 2}, for $\|\lambda(\Delta)\|_1 = \|\Delta\|_1$ to hold, $\lambda(\Delta_+)$ and $\lambda(\Delta_-)$ must have strictly orthogonal supports in the output space $H_d$. Let $\Pi_+$ and $\Pi_-$ be the projectors onto these respective supports, such that $\Pi_+ \Pi_- = 0$.

    Consider the CPTP-map $\mathcal{D}(Y) = \text{Tr}(\Pi_+ Y)|0\rangle\langle0| + \text{Tr}((I-\Pi_+ )Y)|1\rangle\langle1|$ from $\mathcal B(H_d) \rightarrow \mathcal B(H_2)$. 
    Analogously as before, this construction guarantees that $||\mathcal{D}(\lambda(\Delta))||_1=||\lambda(\Delta)||_1$
    
    Consequently, the composed map $\mathcal{E} = \mathcal{D} \circ \lambda$ is a CPTP map into a 2-level system that \textit{still} saturates the bound. By definition of $\mathcal{E}$, its outputs commute. Moreover, since $\mathcal{E}(A)$ and $\mathcal{E}(B)$ are traceless, they have eigenvalues $\{a,-a\}$ and $\{b,-b\}$.
    
    To satisfy the reverse triangle inequality $\|\mathcal{E}(A)\|_1 - \|\mathcal{E}(B)\|_1 = \|\mathcal{E}(A-B)\|_1$,  Corollary \ref{corollary cond 1} implies that $(a-b)b\geq0$. 
    However, we have that
    \begin{equation*}
    \begin{split}
        &a = \mathcal{E}(A)_0= \text{Tr}\big(\Pi_+\lambda(A)\big)= \text{Tr}\big(\Pi_+(\lambda(\Delta)+\lambda(B))\big)=||\Delta_+||_1 +\text{Tr}(\Pi_+\lambda(B))\\
        &b = \mathcal{E}(B)_0= \text{Tr}(\Pi_+\lambda(B)),
    \end{split}
    \end{equation*}
    so $a\geq b$, and then $b \geq 0$.

    Now, using the dual map $\lambda^\dagger$ to define $E=\lambda^\dagger(\Pi_+)$, we have that $\text{Tr}(\Pi_+\lambda(X))=\text{Tr}(EX)$ for any operator $X$. Thus
    \begin{equation}
        \text{Tr}(EP_{\pm})=\text{Tr}(\Pi_+\lambda(P_{\pm}))
        \label{eq: Tr}
    \end{equation}
    and since $\lambda$ is CPTP, $0\leq \Pi_+\leq I \implies 0\leq E \leq I$.
    
    From Lemma \ref{lemma cond 2}, $\lambda$ perfectly preserves the distinguishability of $P_+$ and $P_-$, no trace from the $P_+$ subspace can leak into $\Pi_-$, and vice versa. So $\lambda(P_+)$ lives in the support of $\Pi_+$, then \eqref{eq: Tr} becomes $\text{Tr}(EP_{+})=\text{Tr}(P_{+})$ and $\text{Tr}(EP_{-})=0$. So $E$ must act as the identity on the support of $P_+$ and be orthogonal to $P_-$. Therefore, there is $E_0$ with support on the Kernel of $\Delta$, i.e. $0\leq E_0\leq P_0$, such that $E=P_+ + E_0$

    Which is precisely the condition we desired once
    \begin{equation*}
        b = \text{Tr}(\Pi_+\lambda(B))=\text{Tr}(EB)=\text{Tr}((P_+ + E_0)B) \geq0
    \end{equation*}
    and
    \begin{equation*}
        a = \text{Tr}(\Pi_+\lambda(A))=\text{Tr}(EA)=\text{Tr}((P_+ + E_0)A) \geq b.
    \end{equation*}
       
\end{proof}

\section{Optimization results for bound saturation example}
\label{app: data trivial case}

\begin{figure*}[h]
    \includegraphics[width=0.75\linewidth]{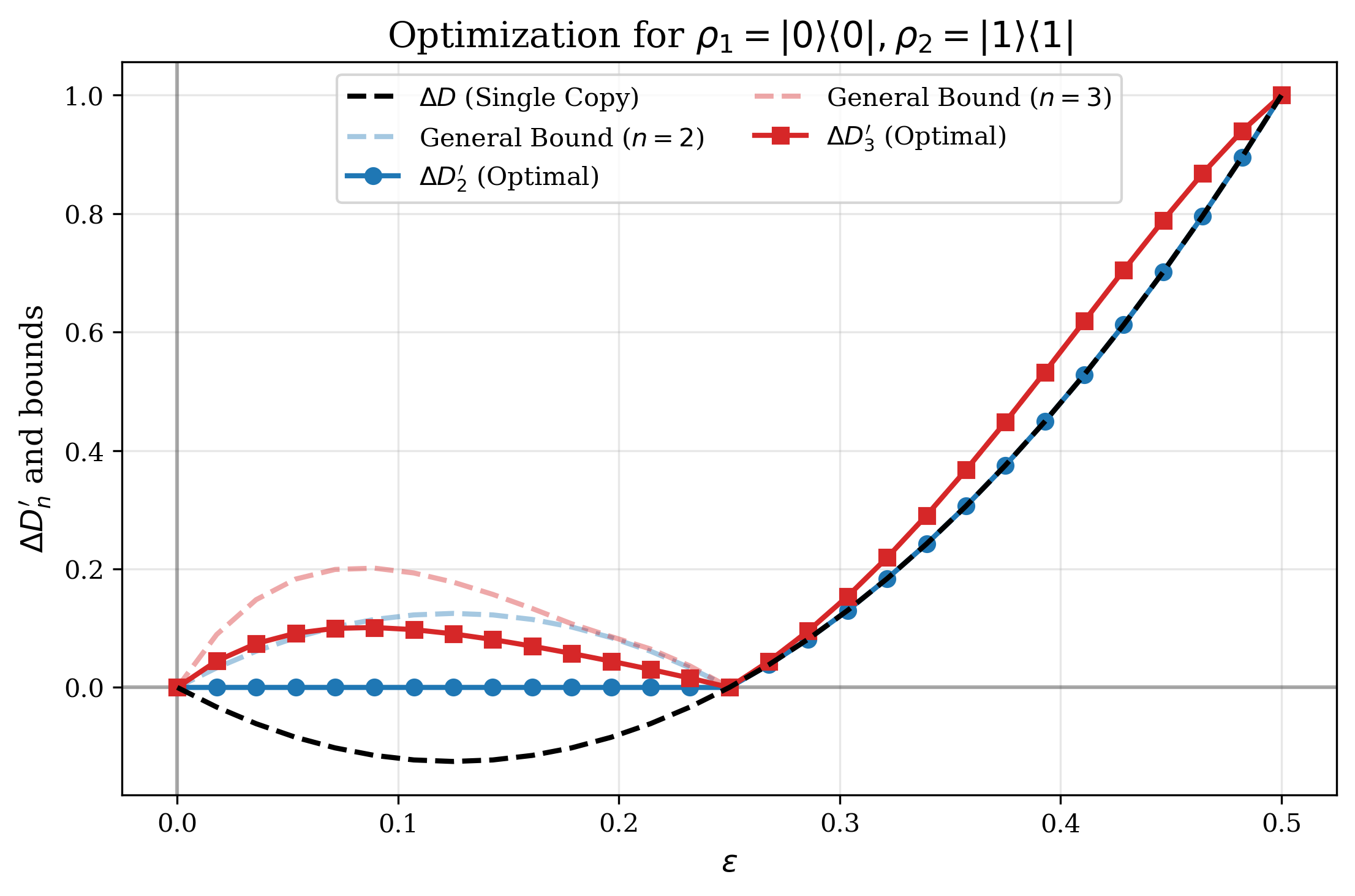}
    \caption{Optimization results for the example introduced in Section~\ref{sec:distillation}. For $\varepsilon > 0.25$, the general bound is attained across the entire region, as expected, for both values of $n$ considered. In contrast, for $\varepsilon < 0.25$, the regime change is observed only for $n=3$, highlighting the multi-copy advantage in activating the regime transition.}
    \label{fig: trivial case}
\end{figure*}

\end{document}